\documentclass[11pt]{article}
\usepackage{a4wide}
\usepackage{hyperref}
\usepackage[usenames,dvipsnames]{color}
\usepackage{amssymb,amsmath,graphics,color,enumerate}
\usepackage{tikz}
\usepackage{xspace}
\usetikzlibrary{arrows,shapes,decorations,automata,backgrounds,petri,patterns} 
\usepackage{todonotes}
\usepackage[utf8x,utf8]{inputenc} 
\usepackage[T1]{fontenc}
\usepackage{booktabs}
\usepackage{array}
\usepackage{mathtools}
\usepackage{algorithm}
\usepackage{caption}
\usepackage[noend]{algpseudocode}
\usepackage[numbers,sort]{natbib}
\usepackage{standalone}
\usepackage{tikz}
\usetikzlibrary{arrows}
\usetikzlibrary{patterns}

\usepackage{epsfig}
\usepackage{subfig}
\usepackage{amsmath,amsfonts,amssymb,epsfig,color,amsthm}
\usepackage{comment}
\usepackage{thmtools}
\usepackage{thm-restate}
\usepackage{mdframed}
\usepackage[absolute]{textpos}

\newtheorem{definition}{Definition}[]
\newtheorem{theorem}{Theorem}[section]
\newtheorem{lemma}[theorem]{Lemma}

\newcommand{\Cc}{{\ensuremath{\mathcal{C}}}}

\newcommand{\Ss}{{\ensuremath{\mathcal{S}}}}
\newcommand{\Tt}{{\ensuremath{\mathcal{T}}}}

\newcommand{\Uu}{{\ensuremath{\mathcal{U}}}}

\newcommand{\Oh}{{\ensuremath{\mathcal{O}}}}

\newcommand{\anonymyze}[1]{}%

\newcommand{\ceil}[1]{\ensuremath{\left\lceil{#1}\right\rceil}}%
\newcommand{\ignore}[1]{}%

\newcommand{\pre}{{\rm pre}}
\newcommand{\post}{{\rm post}}
\newcommand{\low}[1][]{{\rm low#1}}
\newcommand{\mindn}[1][]{{\rm MinDn#1}}
\newcommand{\maxdn}[1][]{{\rm MaxDn#1}}
\newcommand{\maxup}[1][]{{\rm MaxUp#1}}
\newcommand{\DeepestDnCut}[1][]{{\rm DeepestDnCut}}
\newcommand{\DeepestDnCutNoMin}[1][]{{\rm DeepestDnCutNoMin}}
\newcommand{\DeepestDnCutNoMax}[1][]{{\rm DeepestDnCutNoMax}}
\newcommand{\LCA}[1][]{{\rm LCA}}

\DeclareMathOperator{\polylog}{polylog}

\begin{document}

\title{Determining 4-edge-connected components in linear time\thanks{This research is a part of projects that have received funding from the European Research Council (ERC)
		under the European Union's Horizon 2020 research and innovation programme
		(Grant Agreement 714704 (W.~Nadara), 677651 (M.~Smulewicz), and 948057 (M.~Sokołowski)).
		} }

\date{\today}

\author{Wojciech Nadara\thanks{Institute of Informatics, University of Warsaw, Poland (\texttt{w.nadara@mimuw.edu.pl})} \and Mateusz Radecki\thanks{University of Warsaw, Poland (\texttt{mr386052@students.mimuw.edu.pl})} \and Marcin Smulewicz\thanks{Institute of Informatics, University of Warsaw, Poland (\texttt{m.smulewicz@mimuw.edu.pl})} \and Marek Sokołowski\thanks{Institute of Informatics, University of Warsaw, Poland (\texttt{marek.sokolowski@mimuw.edu.pl})}}

\maketitle

\begin{textblock}{20}(0, 13.0)
	\includegraphics[width=40px]{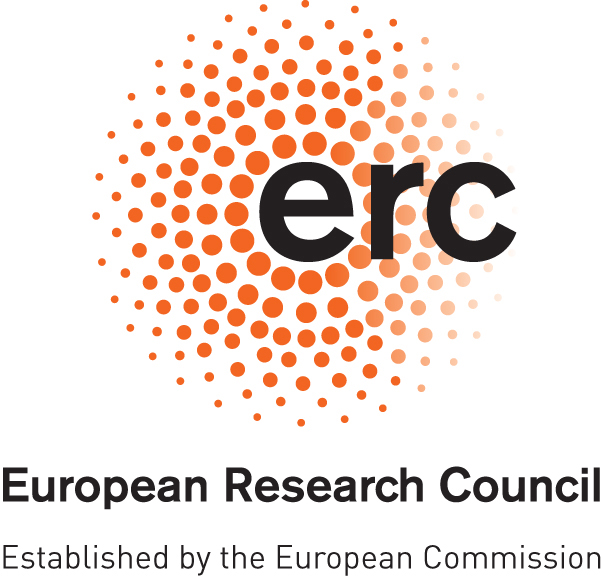}%
\end{textblock}
\begin{textblock}{20}(-0.25, 13.4)
	\includegraphics[width=60px]{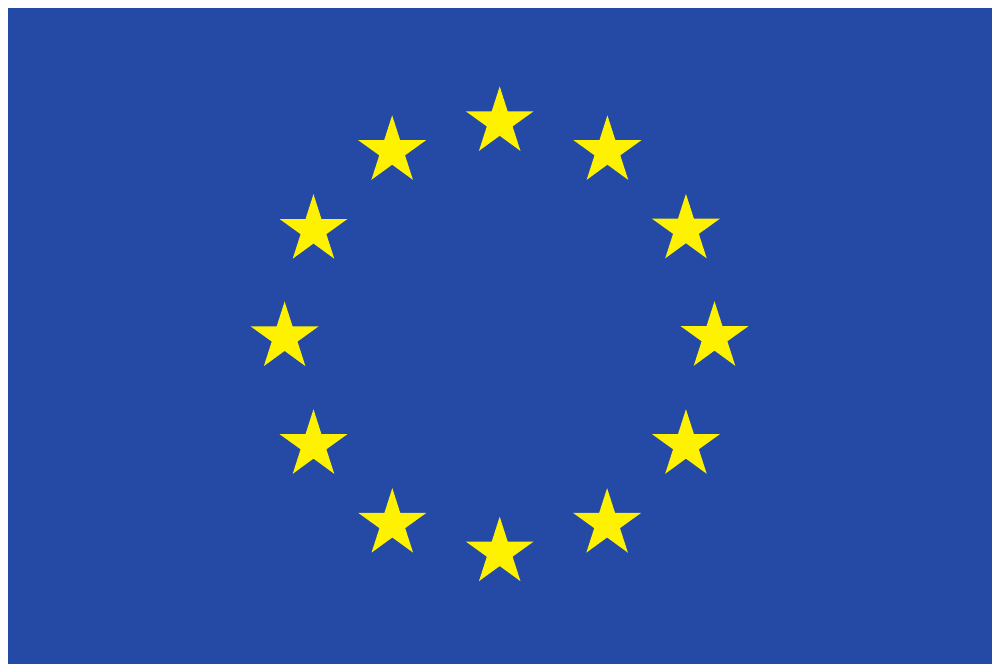}%
\end{textblock}

\begin{abstract}
In this work, we present the first linear time deterministic algorithm computing the 4-edge-connected components of an undirected graph. First, we show an algorithm listing all 3-edge-cuts in a~given 3-edge-connected graph, and then we use the output of this algorithm in order to determine the 4-edge-connected components of the graph.
\end{abstract}

\newpage

\section{Introduction}
\label{sec:introduction}

The connectivity of graphs has always been one of the fundamental concepts of graph theory.
The foremost connectivity notions in the world of undirected graphs are the \emph{$k$-edge-connectedness} and the \emph{$k$-vertex-connectedness}.
Namely, a~graph $G$ is \emph{$k$-edge-connected} for $k \geq 1$ if it is connected, and it remains connected after removing any set of at most $k-1$ edges.
Similarly, $G$ is \emph{$k$-vertex-connected} if it contains at least $k + 1$ vertices, and it remains connected after the removal of any set of at most $k - 1$ vertices.

These notions can be generalized to the graphs that are not well-connected.
Namely, if $H$ is a~maximal $k$-vertex-connected subgraph of $G$, we say that $H$ is a~\emph{$k$-vertex-connected component} of $G$.
The edge-connected variant is, however, defined differently: we say that a~pair of vertices $u, v$ of $G$ is $k$-edge-connected if it is not possible to remove at most $k-1$ edges from $G$ so that $u$ and $v$ end up in different connected components.
This relation of $k$-edge-connectedness happens to be an~equivalence relation; this yields a definition of a~\emph{$k$-edge-connected component} of $G$ as an~equivalence class of the relation.
We remark that the notions of $k$-vertex-connected components and $k$-edge-connected components coincide for $k = 1$ as both simply describe the connected components of~$G$.
However, for $k \geq 2$ these definitions diverge; in particular, for $k \geq 3$ the $k$-edge-connected components of a~graph do not even need to be connected.

%\subsection{Previous work}

There has been a~plethora of research into the algorithms deciding the $k$-vertex- and $k$-edge-connectedness of graphs, and decomposing the graphs into $k$-vertex- or $k$-edge-connected components.
However, while classical, elementary, and efficient algorithms exist for $k = 1$ and $k = 2$, these problems become increasingly more difficult for the larger values of $k$.
In fact, even for $k = 4$, there were no known linear time algorithms to any of the considered problems.
The following description presents the previous work in this area for $k \in \{1, 2, 3, 4\}$, and exhibits the related work for the larger values of $k$:

\paragraph{$\mathbf{k = 1}$.}
Here, the notions of $k$-vertex-connectedness and $k$-edge-connectedness reduce to that of connectivity and the connected components of a~graph. In the static setting, determining the connected components in linear time is trivial.
As a~consequence, more focus is being laid on dynamic algorithms maintaining the connected components of graphs.
In the incremental setting, where the edges can only be added to the dynamic graph, the optimal solution is provided by disjoint-set data structures \cite{FU}, which solve the problem in the amortized $\Oh(\alpha(n))$ time per query, where $\alpha(n)$ denotes the inverse of the Ackermann's fast-growing function. The fully dynamic data structures are also considered
 \cite{Wulff, Sparsification, Cos2, Cos3, Cos4, ThorupDuzo, Cos6, DynConnWorst, Bridge, DynConnSODA17}.

\paragraph{$\mathbf{k = 2}$.}
One step further are the notions of 2-vertex-connectivity (biconnectivity) %not too much focus on vertex connectivity?
and 2-edge-connectivity.
In the static setting, partitioning of a~graph into $2$-vertex-connected or $2$-edge-connected components are classical problems, both solved in linear time by exploiting the properties of the low function~\cite{2ConnStatic}.
The incremental versions of both problems are again solved optimally in the amortized $\Oh(\alpha(n))$ time per query \cite{Westbrook1992}.
Significant research has been done in the dynamic setting as well \cite{DBLP:conf/wads/PengSS19, ThorupDuzo, Sparsification, Bridge, Dyn3EConn, Henzinger1995, Dyn2ConnBack}.

\paragraph{$\mathbf{k = 3}$.} As a next step, we consider 3-vertex-connectivity (triconnectivity) and 3-edge-connectivity.
An~optimal, linear time algorithm detecting the $3$-vertex-connected components was first given by Hopcroft and Tarjan~\cite{TriConn}.
The first linear algorithm for $3$-edge-connectivity was discovered much later by Galil and Italiano~\cite{EdgeToVertex}, where they present a~linear time reduction from the $k$-edge-connectivity problem to $k$-vertex-connectivity for $k \geq 3$, showing that in the static setting, the former problem is the easier of the two.
This was later followed by a~series of works simplifying the solution for $3$-edge-connectivity~\cite{Tsin1, Tsin2, 3E1, Tsin3}.
The incremental setting \cite{3Incr, IncrSPQR} and the dynamic setting \cite{Dyn3EConn, Sparsification} were also considered.

We also mention that in the case of 3-vertex-connectivity, there exists a structure called SPQR-tree which succinctly captures the structure of 2-vertex-cuts in graphs~\cite{SPQR, IncrSPQR}.
Its edge-connectivity analogue also exists, but we defer its introduction to the general setting.
%$k=3$ is also a point, where a question about representing the structure of $(k-1)$-vertex-cuts and $(k-1)$-edge-cuts emerges. The data structure called SPQR-tree was invented in order to capture the structure of 2-vertex-cuts \cite{SPQR, IncrSPQR}.

\paragraph{$\mathbf{k = 4}$.}
We move on to the problems of 4-vertex-connectivity and 4-edge-connectivity.
A~notable result by Kanevsky et al.~\cite{QuadConn} supports maintaining 4-vertex-connected components in incremental graphs, with an~optimal $\Oh(\alpha(n))$ amortized time per query. Their result also yields the solution for static graphs in $\Oh(m + n \alpha(n))$ time complexity.
By applying the result of Galil and Italiano~\cite{EdgeToVertex}, we derive a~static algorithm determining the $4$-edge-connected components in the same time complexity.
This algorithm is optimal for $m = \Omega(n \alpha(n))$.

Another result by Dinitz and Westbrook~\cite{Dinitz1998} supports maintaining the $4$-edge-connected components in the incremental setting. Their algorithm processes any sequence of queries in $\Oh(q + m + n \log n)$ time where $q$ is the number of queries, and $m$ is the total number of inserted edges to the graph.

However, it is striking that the fastest solutions for $4$-edge-connectivity and $4$-vertex-connectivity for static graphs were derived from the on-line algorithms working in the incremental setting.
In particular, no linear time algorithms for $k = 4$ were known before.

%[fully dynamic? anything?]

\paragraph{$\mathbf{k \geq 5}$.}
As a side note, we also present the current knowledge on the general problems of $k$-vertex-connectivity and $k$-edge-connectivity. A~series of results~\cite{Karger, DetMinCut, HenzingerCut, AnotherMinCut} show that it is possible to compute the minimum edge cut of a~graph (i.e., determine the edge-connectivity of a~graph) in near-linear time.
%Karger showed how to compute the edge-connectivity in $\Oh(m \log{n}^3)$ randomized time \cite{Karger}, while Kawarabayashi and Thorup showed how to compute it in $\Oh(m \log{n}^12)$ deterministic time \cite{DetMinCut} later on improved to  $\Oh(m \log{n}^2 \log{\log{n}}^2)$ by Henzinger et al. \cite{HenzingerCut}.
The previously mentioned work by Dinitz and Westbrook~\cite{Dinitz1998} maintains the $k$-edge-connected components of an~incremental graph which is assumed to already have been initialized with a~$(k-1)$-edge-connected graph.
The data structure answers any sequence of on-line queries in $\Oh(q + m + k^2 n \log{(n / k)})$ time, where $q$ is the number of queries, and $m$ is the number of edges in the initial graph.

Gomory and Hu \cite{GomoryHu} proved that for any weighted, undirected graph $G$ there exists a weighted, undirected tree $T$ on the same vertex set such that for any two vertices $s, t \in V(G)$, the value of the minimum $s$-$t$ edge cut in $T$ is equal to the value of the minimum $s$-$t$ edge cut in $G$.
Moreover, such a~tree can be constructed using $n-1$ invocations of the maximum flow algorithm.
In an interesting result by Hariharan et al.~\cite{PartialGomory-Hu}, the decomposition of any graph into $k$-edge-connected components is constructed in $\Oh((m + nk^3) \cdot \polylog(n))$ time, producing a~partial Gomory-Hu tree as its result.

Dinitz et el.~\cite{Cactus} showed that the set of all minimum edge cuts can be succinctly represented with a~\textit{cactus graph}. When the minimum edge cut is odd, this cactus simplifies to a tree (see \cite[Corollary 8]{SimplerCactus}). These results imply that if the size of the minimum cut is odd, then the number of minimum cuts is $\Oh(n)$ and if it is even, then the number of minimum cuts is $\Oh(n^2)$. The structure of $k$-vertex-cuts was also investigated \cite{Longhui}.
%[czy warto wspominać o kwadratowej/liniowej liczbie cutów w zależności od $k$?]\todo{}

%[I recall something incremental for $5$-edge-connectivity by Dinitz]\todo{}

\subsection{Our results}
In this work, we present a linear time, deterministic algorithm partitioning static, undirected graphs into $4$-edge-connected components. Even though the area of the dynamic versions of the algorithms for $k$-edge-connectivity is still thriving, the progress in static variants appears to have plateaued. In particular, both subquadratic algorithms determining the $4$-edge-connected components \cite{QuadConn, Dinitz1998} originate from their dynamic incremental equivalents and are almost thirty years old, yet they did not achieve the optimal linear running time. Hence, our work %algorithm was clearly lacking in the series of results from this area
constitutes the first progress in the static setting of $4$-edge-connectivity in a~long time. As a~side result, our algorithm also produces the tree representation of $3$-edge-cuts as explained in~\cite{SimplerCactus}.

\subsection{Organization of the work}
The paper is organized as follows.
In Section~\ref{Smu1}, we show how to reduce the problem of determining $4$-edge-connected components to the problem of determining $4$-edge-connected components in $3$-edge-connected graphs.
In Section~\ref{sec:randomized_3cut}, we show a~linear time, randomized Monte Carlo algorithm for listing all 3-edge-cuts in 3-edge-connected graphs.
In Section \ref{sec:deterministic_3cut}, we show how to remove the dependency on the randomness in the algorithm from the previous section, producing a~linear time, deterministic algorithm listing all 3-edge-cuts in 3-edge-connected graphs.
Then, in Section \ref{Smu2}, we construct a tree of 3-edge-cuts in a~3-connected graph, given the list of all its 3-edge cuts. This tree is then used to determine the 4-edge-connected components of the graph. 
Finally, in Section \ref{open-problems}, we present open problems related to this work.

\section{Preliminaries}

\paragraph{Graphs.}
  In this work, we consider undirected, connected graphs which may contain self-loops and multiple edges connecting pairs of vertices (i.e., multigraphs).
  The number of vertices of a~graph and the number of its edges are usually denoted $n$ and $m$, respectively.
  
  We use the notions of $k$-edge-connectedness and $k$-edge-connected components defined in Section~\ref{sec:introduction}.
  Moreover, we say that a set of $k$ edges of a~graph forms a~\emph{$k$-edge cut} (or a~\emph{$k$-cut} for simplicity) if the removal of these edges from the graph disconnects it.

\paragraph{DFS trees.}
  Consider a~run of the \emph{depth-first search} algorithm~\cite{DBLP:journals/siamcomp/Tarjan72} on a~connected graph $G$.
  A~\emph{depth-first search tree} (or a~\emph{DFS tree}) is a~spanning tree $\Tt$ of $G$, rooted at the source of the search $r$, containing all the edges traversed by the algorithm.
  After the search is performed, each vertex $v$ is assigned two values: its \emph{preorder} $\pre(v)$ (also called \emph{discovery time} or \emph{arrival time}) and \emph{postorder} $\post(v)$ (also \emph{finishing time} or \emph{departure time}).
  Their definitions are standard~\cite{DBLP:books/daglib/0023376}%\todo{good cite?}
  ; it can be assumed that the values range from $1$ to $2n$ and are pairwise different.

  The edges of $\Tt$ are called \emph{tree edges}, and the remaining edges are called \emph{back edges} or \emph{non-tree edges}.
  In this setup, every back edge $e$ connects two vertices remaining in ancestor-descendant relationship in $\Tt$; moreover, the graph $\Tt + e$ contains exactly one cycle, named the \emph{fundamental cycle} of~$e$.%\todo{need that?}
  For a~vertex $v$ of $G$, we define $\Tt_v$ to be the subtree of $\Tt$ rooted at~$v$; similarly, for a~tree edge $e$ whose deeper endpoint is $v$, we set $\Tt_e = \Tt_v$.

  When a~DFS tree $\Tt$ of $G$ is fixed, it is common to introduce directions to the edges of the graph: all tree edges of $\Tt$ are directed away from the root of $\Tt$, and all back edges are pointed towards the root of $\Tt$.
  Then, $uv$ is a~directed edge (either a~tree or a~back edge) whose \emph{origin} (or \emph{tail}) is $u$, and whose \emph{destination} (or \emph{head}) is $v$.

  For our convenience, we introduce the following definition: a~back edge $e = pq$ \emph{leaps over} a~vertex $v$ of the graph if $p \in \Tt_v$, but $q \notin \Tt_v$; we analogously define \emph{leaping over} a~tree edge $f$.
  
  Moreover, we define a~partial order $\leq_\Tt$ on the vertices of $G$ and the tree edges of $\Tt$ as follows:
  $x \leq_\Tt y$ if the simple path in $\Tt$ connecting the root of $\Tt$ with $y$ also contains $x$.
  Then, $\leq_\Tt$ has one minimal element---the root of $\Tt$---and each maximal element is a~leaf of $\Tt$.
  When the tree $\Tt$ is clear from the context, we may write $\leq$ instead of $\leq_\Tt$.
  Using the precomputed preorder and postorder values in $\Tt$, we can verify if $x \leq_\Tt y$ holds for given $x, y \in V(G) \cup E(\Tt)$ in constant time.

  We use the classical $\low$ function defined by Hopcroft and Tarjan~\cite{2ConnStatic}.
  However, for our purposes it is more convenient to define it as a function $\low : E(\Tt) \to (E \setminus E(\Tt)) \cup \{\perp\}$ such that for a~tree edge $e$, $\low(e)$ is the back edge $uv$ leaping over $e$ minimizing the preorder of its head $v$, breaking ties arbitrarily; or $\perp$, if no such edge exists.
  This function can be computed for all tree edges in time linear with respect to the size of the graph.

\paragraph{Xors.}
For sets $A$ and $B$, by $A \oplus B$ we denote their symmetric difference, which is $(A \cup B) \setminus (A \cap B)$, and we call it a \textit{xor} of $A$ and $B$. Moreover, for non-negative integers $a$ and $b$, by $a \oplus b$ we denote their xor, that is, an integer whose binary representation is a~bitwise symmetric difference of the binary representations of $a$~and $b$. The definitions can be easily generalized to the symmetric differences of multiple sets or integers.

\section{Reduction to the $3$-edge-connected case}
\label{Smu1}
In this section, we present~a way to reduce the problem of building the structure of $4$-connected components to a~set of independent, simpler instances of the problem.
Each produced instance will be a~$3$-edge-connected graph corresponding to a~single $3$-edge-connected component of the original graph.
This transformation of the input will be vital to the correctness of our work since the algorithm described in the following sections assumes the $3$-edge-connectedness of the input graph.
We remark that this is not a~new contribution~\cite{DinitzRatujeDupe}; we present it here for completeness only.

Firstly, each connected component of $G$ can be considered independently.
Similarly, bridges (i.e., $1$-edge cuts) split the given graph into independent $2$-edge-connected components. %, and it is easy to restore shape of $2$-connected components tree if needed.
Moreover, it can be shown that for $k \ge 2$, the family of $k$-edge-connected components of $G$ will not be altered by the removal of the bridges.
Thus, without loss of generality, we assume that $G$ is $2$-edge-connected.

For a~$2$-edge-connected graph $G$, we first build a~structure of its $3$-edge-connected components. The shape of this structure is a~\emph{cactus graph}, asserted by the following theorem:
\begin{theorem} [\cite{DinitzRatujeDupe}] \label{the:2_connected_cactus}
  For a~given $2$-edge-connected graph $G$, there exists an~auxiliary graph $H = (U,F)$ such that
  \begin{itemize}
  	\item each edge $e \in F$ lies on exactly one simple cycle of $H$,
  	\item $U$ is the family of all $3$-edge-connected components of $G$,
  	\item there exists a~bijection $\phi : E' \to F$ where $E' \subseteq E$ is the set of all edges belonging to some $2$-edge cut of $G$, such that for every edge $e = uv \in E'$, its image $\phi(e)$ is an~edge of~$H$ connecting the $3$-edge-connected component containing $u$ with the $3$-edge-connected component containing $v$,
  	\item a~pair of edges $e,f \in E'$ forms a~$2$-edge-cut if and only if $\phi(e)$ and $\phi(f)$ belong to the same cycle of $H$.
  \end{itemize}
\end{theorem}

In order to build the structure from Theorem~\ref{the:2_connected_cactus}, we use the result of Galil and Italiano~\cite{EdgeToVertex} to find all 3-edge-connected components of $G$.
Then, $H$ is defined as a~quotient graph created by identifying the vertices within each 3-edge-connected component.
This can easily be performed in linear time with respect to the size of $G$.

By definition, the partition of the vertices of $G$ into 4-edge-connected components is a~refinement of the partition into 3-edge-connected components.
However, we can no longer restrict our attention to the independent subgraphs induced on 3-edge-connected components (which is what we did in the case of connected components and 2-edge-connected components).
In fact, the $3$-edge-connected components of $G$ may even be disconnected.
In order to handle this problem, we will need to transform $G$ in order to turn each of its 3-edge-connected components into a separate connected component which itself is 3-edge-connected.

Let us now fix $X \subseteq V$---a~$3$-edge-connected component of $G$.
We will now construct a~$3$-edge-connected graph $S = (X, Y)$ whose partition into 4-edge-connected components will be the same as the partition of $X$ into 4-edge-connected components in the original graph.
Initially, we assign to $Y$ the set of all edges of $G[X]$.
Then, for each cycle $C$ of $H$ incident to $X$ ($X \in V(H)$), we add an~additional special edge to $Y$.
Formally, let $\phi(e_1), \phi(e_2) \in F$ be the two edges of the cycle $C$ in $H$ that are incident to $X$.
Then, both edges $e_1, e_2 \in E$ have exactly one endpoint belonging to $X$; denote them $a$ and $b$, respectively.
Then, we add $ab$ to $Y$ as the additional special edge for the cycle $C$.
Intuitively, the new edge simulates a~path in $G$ connecting $a$ and $b$ which is internally disjoint with $X$ and which goes through the 3-edge-connected components around the cycle $C$.

It now turns out that $S$ is $3$-edge-connected and captures the connectivity properties of $X$ in $G$:

%Beside the edges inherited from $G$, set $Y$ contains one special edge for each cycle of $H$ connected to the considered $3$-edge-connected component. More formally, let us consider the cactus structure $H = (U, F)$ from the Theorem \ref{the:2_connected_cactus} and the vertex $x \in U$ that $X$ was mapped into. For each cycle $C$ containing $x$ there are two edges $e$ and $f$ in $C$ adjacent to $x$. Each of their preimages $\phi^{-1}(e)$ and $\phi^{-1}(f)$ have exactly one end belonging to $X$, denote them by $a$ and $b$, respectively. For each such $C$, we put edge $ab$ into $Y$. Intuitively, such edge is added to $Y$ to simulate the fact that 
%vertices $a$ and $b$ are connected by a path going completely outside $X$ 
%around the cycle $C$.

%Precisely, for a cycle $L$ connected to $X$ by vertices $u$ and $v$
%there is a special edge $(u,v)$ in $F$. 
%Intuitively, such edge is added to $S$ to simulate the fact that 
%vertices $u$ and $v$ are connected by a path going completely outside $X$ 
%around the cycle $L$.

\begin{lemma}[\cite{DinitzRatujeDupe}]
  \label{thm:cut_structure}
  The graph $S$ defined above has the following properties:
  \begin{itemize}
    \item it is $3$-edge-connected,
    \item for each $3$-edge-cut $c$ of $G$ dividing $X$ into nonempty parts, there is a $3$-edge-cut $c'$ of $S$ dividing $X$ in the same way,
    \item for each $3$-edge-cut $c'$ of $S$, there is a nonempty set of $3$-edge-cuts of $G$, each dividing $X$ in the same way as $c'$. %; elements of this bunch can be restored from $c'$ by replacing the special edges by preimages of arbitrary edges from the corresponding cycles.
  \end{itemize}
\end{lemma}
\begin{figure}[h]
	\centering
	\includegraphics[scale=0.6]{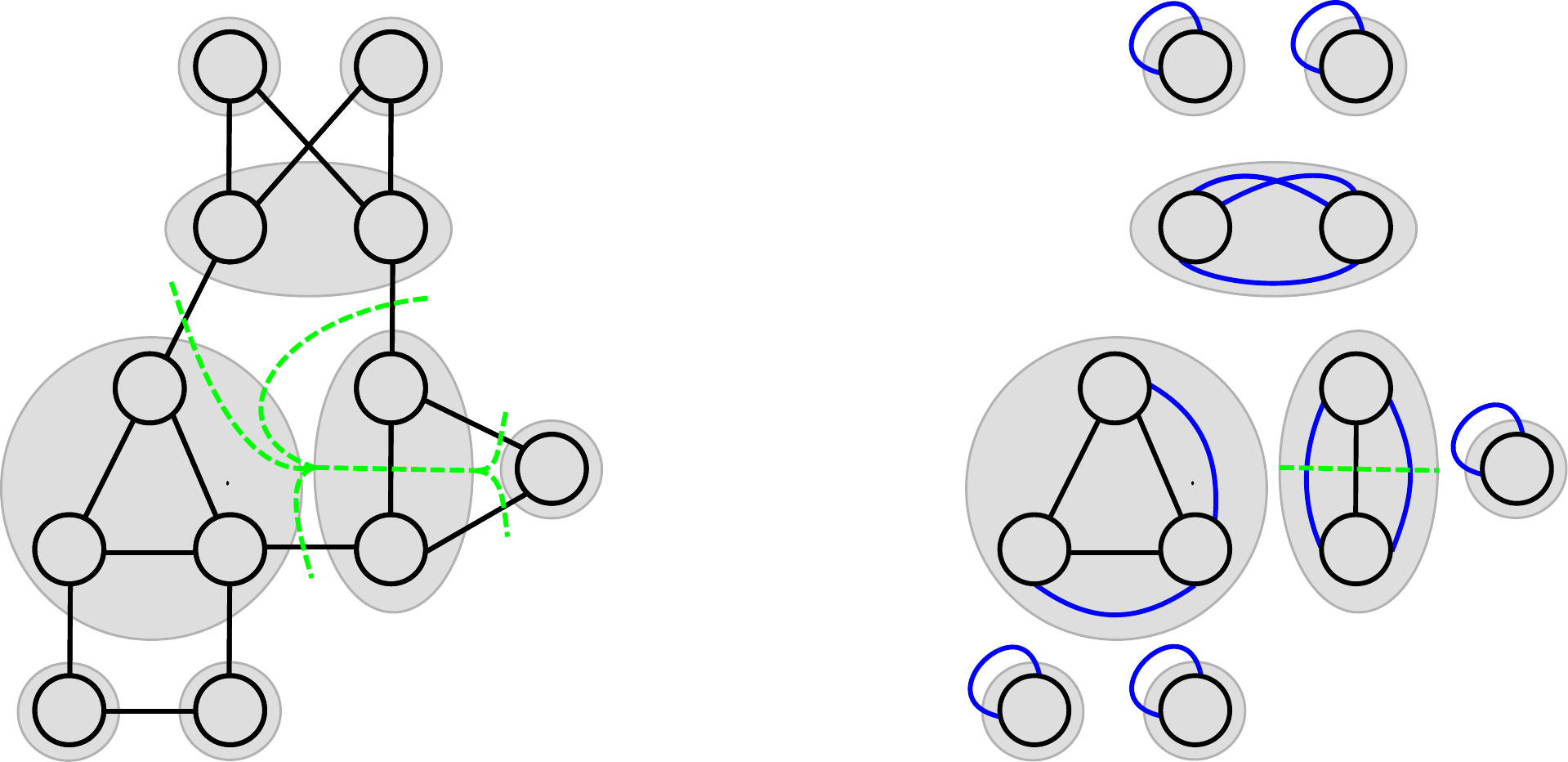}
	\caption{A $2$-edge-connected graph with marked $3$-edge-connected components (left) and corresponding graphs $S$ for each $X$ with blue special edges (right). Green dashed edges show an example of all $3$-edge-cuts (left) corresponding to one $3$-edge-cut (right).}
	\label{fig:treerec_example}
\end{figure}

By Lemma~\ref{thm:cut_structure}, each $4$-edge-connected component $Y \subseteq X$ of $S$ is also a~$4$-edge-connected component of $G$.
Therefore, the reduction is sound.

Given the structure of $3$-edge-connected components of $G$, it is easy to construct graphs $S$ for each particular $X$ in total $\Oh(n + m)$ time, hence the only remaining part is to construct the structure of $4$-edge-connected components for each $S$ independently. 

Observe that the total number of special edges added to all $3$-edge-connected components is equal to the total length of the cycles in $H$.
Thus, the reduction can easily be performed in time linear with respect to the size of $G$.
%Thus, assuming we can construct such structure for one $3$-edge-connected component in linear time, the total time is also linear.
As a result, without loss of generality, we can assume that the given graph $G$ is $3$-edge-connected.

\section{Simple randomized algorithm}
\label{sec:randomized_3cut}
In this section, we are going to describe a randomized linear time algorithm listing 3-edge-cuts in 3-edge-connected graphs.
In particular, the existence of this algorithm will imply that the number of 3-edge-cuts in any 3-edge-connected graph is at most linear. It has no major advantages over the algorithm presented in the succeeding section, but it is significantly simpler and it already contains most of the core ideas. Thus, it serves as a good intermediate step in the explanation.

We begin with the description of some auxiliary data structures.

\begin{theorem}[\cite{DBLP:journals/jcss/GabowT85}]\label{FU}
  \label{thm:linear_fu}
  There exists a~data structure for the disjoint set union problem which, when initialized with an~undirected tree $\Tt$ (the ``union tree'') on $n$ vertices, creates $n$ singleton sets.
  After the initialization, the data structure accepts the following queries in any order:
  \begin{itemize}
    \item $\mathsf{find}(x)$: returns the index of the set containing $x$,
    \item $\mathsf{union}(x, y)$: if $x$ and $y$ are in different sets, then an arbitrary one of them is replaced with their sum and the other one with the empty set.
      This query can only be issued if $xy$ is an edge of $\Tt$.
  \end{itemize}
  The data structure executes any sequence of $q$ queries in total $\Oh(n + q)$ time.
\end{theorem}

\begin{lemma}\label{enrich-parent}
  \label{lem:linear_fu_with_lowest}
%WN: Jak nam bedzie bardzo miejsca brakowac, to mozna sprobowac zmergowac ten lemat z tym cytowanym twierdzeniem
It is possible to enrich the data structure from Theorem \ref{FU}, so that after rooting it at an arbitrary vertex, we are able to answer the following query in constant time:
	\begin{itemize}
		\item $\mathsf{lowest}(x)$: returns the smallest vertex of $S$ with respect to $\le_{\Tt}$, where $S$ is the set containing $x$.
	\end{itemize}
\end{lemma}

Note that each set induces a connected subgraph of $\Tt$, hence the smallest vertex of $S$ is well-defined.

\begin{proof}
	We enrich our structure with information about the lowest common ancestor of all elements of such a~set. Call this $\LCA(S)$ for a set $S$. Whenever we merge two sets $S$ and $R$, it holds that one of $\LCA(S)$ and $\LCA(R)$ is the ancestor of the other. Moreover, this ancestor can be determined in constant time. Hence, $\LCA(S \cup R)$ can be determined in constant time. 
\end{proof}

\begin{theorem}
  \label{thm:min_edge_path}
  There exists a~deterministic algorithm that takes as input:
  \begin{itemize}
    \item an~undirected, unrooted tree $\Tt$ with $n$ vertices,
    \item $p$ weighted paths $P_1, P_2, \dots, P_p$ in the tree,
      where the path $P_i$ ($1 \leq i \leq p$) is the unique path between vertices $u_i$, $v_i$
      and has weight $w_i \in \{0, 1, \dots, C\}$,
    \item and a~positive integer $k$,
  \end{itemize}
  and for each edge $e$ of the tree returns the indices of $k$ paths with the lowest weight containing~$e$,
    breaking ties arbitrarily; if $e$ is a~part of fewer than $k$ paths, all such paths are returned.
  The time complexity of the algorithm is $\Oh(nk + p + C)$.
  \begin{proof}
    Since weights $w_i$ belong to the set $\{0, 1, \ldots, C\}$ we can sort paths $P_1, \ldots, P_p$ with respect to their weights in time $\Oh(p + C)$. Hence, from this point on, we assume that it has already been done and that $w_1 \le w_2 \le \ldots \le w_p$.
    
    We now process these paths in that order. We create the data structure from Lemma \ref{enrich-parent} and initialize it with $\Tt$. For each edge $e \in E(\Tt)$, we are going to maintain a~set $R_e$, which at the end will be the set of the indices of the desired paths for edge $e$.
    
    We root $\Tt$ in an arbitrary vertex, which allows us to define the function $\mathsf{parent}(v)$ mapping $v$ to its parent in $\Tt$.
    
    We define an auxiliary function \textsc{Go}$(u, v, i)$ with the following pseudocode:
    
    \begin{algorithm}[H]
    	\begin{algorithmic}
    		
    		\caption{Updating information with the path $P_i$ on its part from $u_i$ to $\LCA(u_i, v_i)$} \label{alg:min_path_oracle_halfpath}
    		\Function {Go} {$u, v, i$}
    		\State $u \gets \mathsf{lowest}(u)$
    		\While {$u$ is not an ancestor of $v$}
    		\State $e \gets (u, \mathsf{parent}(u))$
    		\State $R_e \gets R_e \cup \{i\}$
    		\If {$|R_e| = k$} $\mathsf{union}(u, \mathsf{parent}(u))$ \EndIf
    		\State $u \gets \mathsf{lowest}(\mathsf{parent}(u))$
    		\EndWhile
    		\EndFunction
    		
    	\end{algorithmic}
    \end{algorithm}
    
    Whenever we process path $P_i$, we execute \textsc{Go}$(u_i, v_i, i)$ and \textsc{Go}$(v_i, u_i, i)$. 
    
    We maintain an invariant that after processing the paths $P_1, \ldots, P_i$, sets $R_e$ are the sets of indices of $k$ first paths containing $e$; or all of them, if there are less than $k$ paths among $P_1, \ldots, P_i$ containing $e$ as an edge. Moreover, $\mathsf{union}(u, v)$ is executed as soon as the size of $R_{uv}$ becomes equal to $k$. The reader is encouraged to think of the $\mathsf{union}$ as the contraction of the edge $uv$ in $\Tt$. A~contracted subgraph is identified with the lowest common ancestor of all its vertices (i.e., the result of $\mathsf{lowest(x)}$ for any $x$ from such a~subgraph). The function \textsc{Go}$(u, v, i)$ traverses all non-contracted edges on the path from $u$ to $l$, where $l$ is the lowest common ancestor of $u$ and $v$, and adds the index $i$ to the sets $R_e$ for all traversed edges $e$ (and contracts them if necessary). The total size of all sets $R_e$ will never exceed $nk$, hence the total number of iterations of the \textsc{while} loop throughout all executions of \textsc{Go} will not exceed $nk$ as well. We infer that the time complexity of this algorithm is $\Oh(nk + p + C)$.
  \end{proof}
\end{theorem}

We proceed to the description of our randomized algorithm.
Let us choose an arbitrary vertex $r$ of the graph, and perform a depth-first search from $r$. Let $\Tt$ be the resulting DFS tree. We are now going to define a \emph{hashing function} $H : E \to \mathcal{P}(E)$, where $\mathcal{P}(E)$ denotes the powerset of $E$; the value $H(e)$ will be called \emph{a~hash} of $e$. If $e \not\in \Tt$, then we define $H(e) = \{e\}$. Otherwise, we take $H(e)$ as the set of non-tree edges leaping over $e$. Let us note that the set $C_e$ of edges $f$ such that $e \in H(f)$ forms a cycle---the fundamental cycle of $e$.

\begin{lemma}\label{xoring-edges}
	For a connected graph $G = (V, E)$ and a subset $A$ of its edges, the graph $G' \coloneqq G  - A$ is disconnected if and only if there is a nonempty subset $B \subseteq A$ such that xor of the hashes of the edges in $B$ is an empty set.
\end{lemma}
\begin{proof}
	Let us start with proving that if $G' = (V, E \setminus A)$ is disconnected, then there is a~nonempty subset $B \subseteq A$ such that xor of hashes of edges from $B$ is the empty set.
	
	If $G'$ is disconnected, then we can partition $V$ into two nonempty sets $L$ and $R$ such that all edges between $L$ and $R$ belong to $A$.
	Let $B$ be the set of the edges between $L$ and $R$.
	We claim that xor of the hashes of the edges from $B$ is the empty set.
	Denote that xor by $X$.
	Let us take any edge $e$.
	Since $C_e$ is a cycle, it contains an even number of edges from $B$.
	Hence $e \not\in X$, which proves that $X = \emptyset$.
	
	Now, let us prove that if there exists a subset $B \subseteq A$ such that xor of the hashes of its edges is the empty set, then $G'$ is disconnected.
	Let us color vertices of $G$ red and blue such that two vertices connected by a~tree edge of $\Tt$ have different colors if and only if this tree edge belongs to $B$.
	Since $\Tt$ is a tree, this coloring always exists, and is unique up to swapping the colors. $B$ has to contain at least one tree edge; otherwise xor of the hashes of its edges would clearly be nonempty.
	 Hence, in such a~coloring, there are vertices of both colors.
	 Let $L$ and $R$ denote the nonempty sets of vertices colored blue and red, respectively.
	 We claim that there is no edge in $E \setminus B$ connecting $L$ and $R$, which in turn will conclude the proof.
	 For the sake of contradiction, assume that such an edge $e \in E \setminus B$ exists.
	 It clearly cannot be a tree edge based on how we defined our coloring.
	 If $e$ is a non-tree edge connecting the vertices of different colors, it has to leap over an odd number of tree edges in $B$.
	 These are exactly the edges of $B \cap E(\Tt)$ that contain $e$ in their hashes.
	 Since xor of the hashes of the edges from $B$ is assumed to be the empty set, $e$ has to be contained in $B$ as well, because $e$ is the only non-tree edge that contains $e$ in its hash---a contradiction.
\end{proof}
%\todo{WN: Wspomniec cos o matroidzie kograficznym}

Since our graph has no 1-edge-cuts, there are no edges whose hashes are empty sets; and since our graph has no 2-edge-cuts, no two edges have equal hashes.
Hence, removing a set of three edges disconnects a~3-edge-connected graph if and only if xor of the hashes of all of them is the empty set.
Moreover, after removing some 3-edge-cut, the graph disconnects into exactly two components, and no removed edge connects vertices within one component.

As storing hashes as sets of edges would lead to inefficient computations, we will define \textit{compressed hashes}. We express them as $b$-bit numbers, where $b = \ceil{3\log_2(m)}$, i.e., as a~function $CH : E \to \{0, \ldots, 2^b-1\}$. For each non-tree edge $e$, we draw $CH(e)$ randomly and uniformly from the set of $b$-bit numbers. For a tree edge $e$, we define its compressed hash $CH(e)$ as the xor of the compressed hashes of the edges in its hash, i.e. $CH(e) = \bigoplus_{f \in H(e)} CH(f)$. Note that since $b = \Oh(\log m)$, we can perform arithmetic operations on compressed hashes in constant time. However, this comes at a~cost of allowing the collisions of compressed hashes: it might happen (although with a~very low probability) that two sets of edges have equal xors of their compressed hashes, but unequal xors of their original hashes.

For the ease of exposition, in the description of this algorithm we will use hash tables, which only guarantee the expected average constant time complexity per lookup.%\todo{WN: reference}
Because of that, the described algorithm will have expected linear instead of worst-case linear running time.
However, there is a way to remove the dependence on hash tables and keep the running time worst-case linear.
This will be explained at the end of this section (Subsection~\ref{ssec:removing_hash_tables}).
For the time being, we create a~hash table $M : \{0, \ldots, 2^b-1\} \to E \cup \{\perp\}$, which for any $b$-bit number $x$ returns an edge whose compressed hash is equal to $x$; or $\perp$ if no such edge exists. In the unlikely event that there are multiple edges whose compressed hashes are equal to $x$, $M$ returns any of them. 

We will now categorize 3-edge-cuts based on the number of tree edges they contain, and explain how to handle each case. Throughout the case analysis, we will use the following fact multiple times:

\begin{lemma}
  \label{lem:randomized_unique_remaining_edge}
  Given two edges $e$ and $f$ of some 3-edge-cut in a 3-edge-connected graph, the remaining edge is uniquely identified by its hash: $H(e) \oplus H(f)$.
  \begin{proof}
    If $\{e, f, g\}$ is a~3-edge-cut, then $H(e) \oplus H(f) \oplus H(g) = \varnothing$, so $H(g) = H(e) \oplus H(f)$.
  \end{proof}
\end{lemma}

\subsection{Zero tree edges}\label{0t}
\label{ssec:randomized_0tree}
As $\Tt$ is a spanning connected subgraph, 3-edge-cuts not containing any tree edges simply do not exist.
\subsection{One tree edge}\label{1t}
\label{ssec:randomized_1tree}
Let $e$ be an edge of $\Tt$ that is the only tree edge of some 3-edge-cut. The two resulting connected components after the removal of such a cut are $\Tt_e$ and $\Tt \setminus \Tt_e$ (we are slightly abusing the notation here; we actually mean that the connected components are the subgraphs induced by the vertex sets of $\Tt_e$ and $\Tt \setminus \Tt_e$, respectively). Since $e$ is not a bridge, $\low(e)$ is well-defined and it connects $\Tt_e$ with $\Tt \setminus \Tt_e$, so it belongs to this cut as well. This uniquely determines the third edge of this cut as it has to be the unique edge whose hash is $H(e) \oplus H(\low(e))$.

In order to detect all such cuts, we iterate over all tree edges $e$ and for each of them we look up in $M$ whether there exists an edge whose compressed hash is $CH(e) \oplus CH(\low(e))$. If it exists and if it turns out to be a non-tree edge, then we call it $g$ and output the triple $\{e, \low(e), g\}$ as a 3-edge-cut.
\subsection{Two tree edges}\label{2t}

Let $e$ and $f$ be the edges of $\Tt$ that form a 3-edge-cut together with some back edge $g$.

\begin{lemma}
	The edges $e$ and $f$ are comparable with respect to $\le_{\Tt}$.
\end{lemma}
\begin{proof} Assume otherwise. Then $\Tt_e$ and $\Tt_f$ are disjoint, and $L \coloneqq \Tt_e \cup \Tt_f$ and $R \coloneqq \Tt \setminus L$ are two connected components resulting from the removal of the 3-edge-cut containing $e$ and $f$. However, $\low(e)$ and $\low(f)$ are well-defined different edges connecting $L$ and $R$, so any cut between $L$ and $R$ would need to contain $e, f, \low(e)$ and $\low(f)$ at the same time---a~contradiction. This proves that $e$ and $f$ are comparable.
	\end{proof}

Without loss of generality, assume that $e <_\Tt f$. Then, the two resulting connected components
 are $L \coloneqq \Tt_e \setminus \Tt_f$ and $R \coloneqq \Tt \setminus L$. The remaining back edge $g$ has to connect $L$ and $R$. We will distinguish two cases, differing in the location of the remaining back edge:
 \textit{two tree edges, lower case}, where $g$ connects $L$ with $\Tt_f$, and \textit{two tree edges, upper case}, where $g$ connects $L$ with $\Tt \setminus \Tt_e$.
 
\begin{figure}[h]
	\centering
	\includegraphics[scale=0.4]{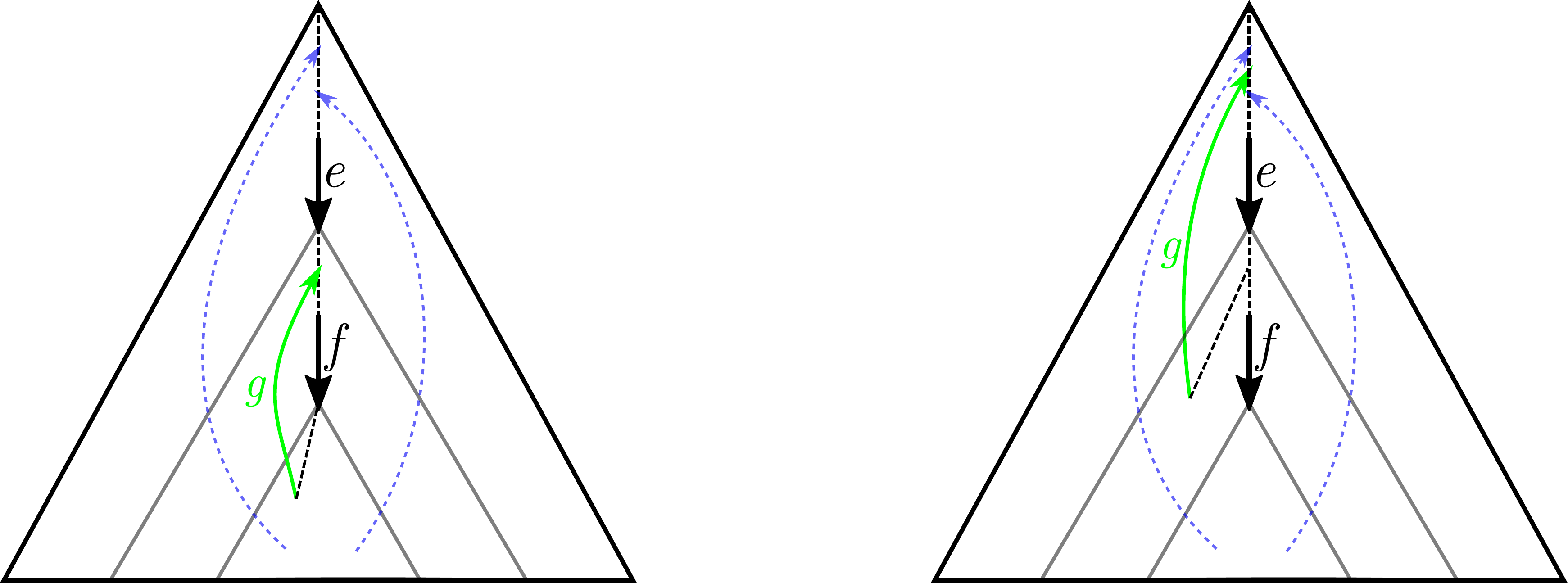}
	\caption{The settings in \ref{2tl} and \ref{2tu}. Cuts that we are looking for are formed by edges $\{e, f, g\}$ in both cases.}
	\label{fig:randomized-two-tree-edge-cases}
\end{figure} 
 
\subsubsection{Two tree edges, lower case}\label{2tl}
\label{ssec:randomized_2tree_lo}

Let $A$ be the set of back edges between $\Tt_f$ and $\Tt \setminus \Tt_f$. It consists of the edge $g$, which connects $L$ with $\Tt_f$, and of several edges connecting $\Tt_f$ with $\Tt \setminus \Tt_e$. Let $B$ be the set of heads of edges from $A$ (we remind that back edges are directed towards the root $r$). All elements of $B$ lie on the path from $f$ to the root $r$, and the head of $g$ is the deepest element of $B$.

Based on this observation, we are going to use the data structure from Theorem \ref{thm:min_edge_path}. We initialize an instance of it with the tree $\Tt$, $k=1$, $C=2n$ and paths $\{P_e \ | \ e$ is a back edge$\}$, where for each non-tree edge $e=xy$, where $y$ is its head, we create one input path $P_e$ from $x$ to $y$, and we set $w_e \coloneqq 2n - \pre(y)$. This data structure lets us for each tree edge $f$ determine the back edge leaping over $f$ with the biggest value of preorder of its head. We call this edge $\maxup(f)$. This back edge is the only candidate for the edge $g$ for a~fixed edge $f$.

Hence, we can find all such cuts by firstly initializing the data structure, and then iterating over all tree edges $f$.
For each $f$, we take $g=\maxup(f)$. Knowing $f$ and $g$, we can look up in $M$ whether there exists an edge whose compressed hash is $CH(f) \oplus CH(g)$ (Lemma~\ref{lem:randomized_unique_remaining_edge}). If it exists and if it is a tree edge, then we call it $e$ and output the triple $\{e, f, g\}$ as a 3-edge-cut.

\subsubsection{Two tree edges, upper case}\label{2tu}
\label{ssec:randomized_2tree_hi}

Let $A$ be the set of non-tree edges between $\Tt_e$ and $\Tt \setminus \Tt_e$. It consists of the back edge $g$, which connects $L$ with $\Tt \setminus \Tt_e$, and of several edges connecting $\Tt_f$ with $\Tt \setminus \Tt_e$. Let $B$ be the set of the tails of the edges from $A$.
Let $v \in L$ be the tail of $g$. As $v \notin \Tt_f$, either for all vertices $u \in \Tt_f$ it holds that $\pre(v) < \pre(u)$, or for all vertices $u \in \Tt_f$ it holds that $\pre(u) < \pre(v)$.
In the first case, $v$ is the vertex with the smallest preorder which is a~tail of some edge leaping over $e$; while in the second case $v$ is the analogous vertex with largest preorder.

Based on this observation, we are going to use the data structure from Theorem \ref{thm:min_edge_path} again.
We initialize one instance of it with the tree $\Tt$, $k=1$, $C=2n$ and the paths $\{P_e \,\mid\, e\text{ is a back edge}\}$, where for each back edge $e=xy$, we create one input path $P_e$ from $x$ to $y$, and set $w_e \coloneqq \pre(x)$.
We also initialize another instance of this data structure in the same way, with the only difference that we set $w_e \coloneqq 2n - \pre(x)$ instead.
The first instance lets us for each tree edge $e$ determine the edge $\mindn(e)$ leaping over $e$ with the smallest preorder of its tail; while the second instance determines the analogous edge $\maxdn(e)$, only with the largest preorder.
By our considerations above, if any~desired $\{e, f, g\}$ cut exists, then $g \in \{\mindn(e), \maxdn(e)\}$.

Hence, we can find all such cuts by firstly initializing both instances of the data structure. Then, we iterate over all tree edges $e$, and for each $e$ we check both candidates for $g$. If our hash table contains an~edge whose compressed hash is $CH(e) \oplus CH(g)$ and it is a tree edge, then we call it $f$ and output the triple $\{e, f, g\}$ as a 3-edge-cut.

\subsection{Three tree edges}\label{3t}
\label{ssec:randomized_3tree}

Solving this case in a way similar to the previous cases seems intractable. We consider this subsection, together with the time analysis following it, as one of the key ideas of this work.

In this case, we assume that no non-tree edges belong to the cut.
Therefore, we can contract all of them simultaneously and recursively list all 3-edge-cuts in the resulting graph.
Since 3-edge-cuts in the contracted graph exactly correspond to the 3-edge-cuts consisting solely of tree edges in the original graph, this reduction is sound.
The contraction is performed in $\Oh(n + m)$ time by identifying the vertices within the connected components of $G - E(\Tt)$. Let $G'$ be the graph after these contractions. Since $G$ is 3-edge-connected, $G'$ has to be 3-edge connected as well, so the assumption about the input graph being 3-edge-connected is preserved. We do not modify the value of $b$ in the subsequent recursive calls, even though the value of $m$ will decrease.

\subsection{Time analysis}
\label{ssec:randomized_time_analysis}
We will now compute the expected time complexity of our algorithm. The subroutines detecting the cuts with at most two tree edges clearly take expected linear time. The only nontrivial part of the time analysis is the recursion in Subsection \ref{3t}. Let $T(m)$ denote the maximum expected time our algorithm needs to solve any graph with at most $m$ edges.
Since our graph is 3-edge-connected, the degree of each vertex in $G$ is at least three, so $|E(G)| \ge \frac{3}{2} |V(G)| \ge \frac{3}{2} |E(G')| \Rightarrow |E(G')| \le \frac{2}{3} |E(G)|$. The bound on $T(m)$ follows:
$$T(m) \le \Oh(m) + T\left(\frac{2}{3}m\right).$$
The solution to this recurrence is $T(m) = \Oh(m)$, hence the whole algorithm runs in expected linear time.

\subsection{Correctness analysis}
In this subsection, we are going to prove that our algorithm works with sufficiently high probability. The only reason it can output wrong result comes from the compression of hashes---if we used their uncompressed version instead, the algorithm would clearly be correct.

%Note that if some hash is empty then its compressed version is clearly equal to zero. However, if it is not empty, then it is equal to zero with $2^{-b}$ probability. Since our input graph is 3-edge-connected, there is no pair of edges with equal hashes, hence xor of any two of them is not empty, so for any fixed pair of edges, probability that their compressed hashes are equal is equal to $2^{-b}$, therefore the probability that compressed hashes of edges are not unique can be bounded from above by ${m \choose 2} 2^{-b}$. Note that in the event where compressed hashes of edges are unique there is at most one edge that could be returned for each key in our hash table.

The maximum number of queries to our hash table is linear in terms of $n$, which follows from a similar argument to the one presented in Subsection~\ref{ssec:randomized_time_analysis}.
Hence, there exists some absolute constant $c$ such that the number of queries is bounded by $cn$.
For each query with value $q$, and for each edge whose hash is not equal to the hash whose compressed version we ask about, there is $2^{-b}$ probability that compressed version of this hash is equal to $q$.
Hence, the probability that we ever get a~false positive is bounded from above by $cnm2^{-b}$.
If we never get a~false positive, then our algorithm returns the correct output. Since $b = \ceil{3 \log_2(m)}$ and $n \le m$, we infer that $cnm2^{-b} \le \frac{c}{m}$.
Therefore, our algorithm works correctly with probability at least $1 - \frac{c}{m}$.

\subsection{Removing hash tables}
\label{ssec:removing_hash_tables}
As mentioned earlier, there is a way to avoid using hash tables in our algorithm. Since the values of compressed hashes of edges are polynomial in $m$ (at most $2m^3$), we are able to sort them in $\Oh(m)$ time using the radix sort.
%\todo{WN: reference?}
Now, instead of requiring a data structure that allows us to query the existence of an edge with a particular value of its compressed hash, we create an object representing such query, which we are going to answer in the future in an offline manner. After gathering all such objects, we may sort them together with edges, where as keys for comparisons we use values of compressed hashes for edges and values of required compressed hashes for queries. All the keys for comparisons are polynomial, so we are able to sort them together in $\Oh(m)$ time using the radix sort. After the sorting, we are able to easily answer all queries offline.

This concludes the description of the randomized variant of the algorithm.

\section{Deterministic algorithm}
\label{sec:deterministic_3cut}
%\todo{pack algorithms to figures}

Having established a~linear time randomized algorithm producing 3-edge-cuts in $3$-edge-connected graphs, we will now determinize it by designing deterministic implementations of the subroutines for each of the cases considered in the randomized variant of the algorithm.
Three cases need to be derandomized: ``one tree edge'' (Subsection~\ref{ssec:randomized_1tree}), ``two tree edges, lower case'' (Subsection~\ref{ssec:randomized_2tree_lo}), and ``two tree edges, upper case'' (Subsection~\ref{ssec:randomized_2tree_hi}).
In the following description, we cannot use compressed hashes anymore: the compression is a~random process which inevitably results in false positives.
Instead, we will exploit additional properties of $3$-edge-cuts in order to produce an~efficient deterministic implementation of the algorithm.

Recall that in the description of the randomized implementation of the algorithm, we defined the values $\low(e)$, $\maxup(e)$, $\mindn(e)$, and $\maxdn(e)$ for any tree edge $e$.
The values $\low(e)$ and $\maxup(e)$ were defined as the back edges $x = uv$ leaping over $e$ whose head $v$ has the smallest possible preorder in~$\Tt$, or the largest possible preorder in~$\Tt$, respectively, breaking ties arbitrarily.
The values $\mindn(e)$ and $\maxdn(e)$ were defined analogously, only that we chose edges with the smallest possible preorder of the source $u$, and the largest possible preorder of $u$, respectively.%\todo{move to random}

For the deterministic variant of the algorithm, we generalize these notions: we define $\low[1](e)$, $\low[2](e)$, and $\low[3](e)$ as the three back edges leaping over $e$ with the minimal preorders of their targets; in particular, we set $\low[1](e) := \low(e)$.
We analogously define $\maxup[1](e)$, $\maxup[2](e)$, $\mindn[1](e)$, $\mindn[2](e)$, $\maxdn[1](e)$, and $\maxdn[2](e)$.
We remark that $\low[3](e)$ might not exist if there are fewer than three edges leaping over $e$; in this case, we put $\low[3](e) := \bot$.
However, all the other values must exist---otherwise, at most one edge would leap over $e$, which would mean that this edge, together with $e$, would form a~$2$-edge-cut of $G$ of cardinality~$2$.

All the values defined above can be computed for every tree edge $e$ in linear time with respect to the size of $G$: the values $\low[1](e)$, $\low[2](e)$, and $\low[3](e)$ are computed in a single depth-first search pass along $\Tt$ in the same way as the original low function is computed,%\todo{citation}
only that three lowest back edges are computed instead of just one.
The generalizations of $\maxup$, $\mindn$, and $\maxdn$ are determined in the same way as in Section~\ref{sec:randomized_3cut}, but $k = 2$ is passed to the algorithm from Theorem~\ref{thm:min_edge_path} instead of $k = 1$ so that two best back edges of each kind are computed instead of just one.
Hence, in the following description, we will assume that all the values above have already been computed.

\subsection{One tree edge}
\label{ssec:deterministic_1edge}

Recall that in this case, we are to find all $3$-edge-cuts intersecting a~fixed depth-first search tree $\Tt$ in a~single edge.
This is fairly straightforward: if some $3$-edge-cut $C$ contains exactly one tree edge $e$ of $\Tt$, then the border of the cut is exactly $\Tt_e$; hence, this cut must include all back edges leaping over $e$.
Therefore, $e$ belongs to a~$3$-edge-cut of this kind if and only if $\low[3](e) = \bot$, i.e. if there only exist two back edges leaping over $e$; in this case, $e$ and these two edges form a $3$-edge-cut.
Since this check can be easily performed in $\Oh(n)$ time for all tree edges in $\Tt$, the whole subroutine runs in $\Oh(n + m)$ time complexity.

\subsection{Two tree edges, lower case}
\label{ssec:deterministic_2edges_down}

Recall that this case requires us to find all $3$-edge-cuts intersecting a~fixed depth-first search tree $\Tt$ in two edges, say $e$ and $f$, such that $e < f$ and the remaining back edge $g$ connects $\Tt_f$ with $\Tt_e \setminus \Tt_f$ (Figure \ref{fig:randomized-two-tree-edge-cases}).
%\todo{figure from random}
Hence, $g$ is the only back edge connecting $\Tt_f$ with $\Tt_e \setminus \Tt_f$, no back edges connect $\Tt_e \setminus \Tt_f$ with $\Tt \setminus \Tt_e$, but there may be multiple back edges connecting $\Tt_f$ with $\Tt \setminus \Tt_e$.

We shall exploit the fact that in a~valid $3$-edge-cut of this kind, there are no back edges originating from $\Tt_e \setminus \Tt_f$ and leaping over $e$.
For a~fixed edge $e$, this severely limits the set of possible tree edges $f$ below $e$ in $\Tt$, since all back edges leaping over $e$ must originate from $\Tt_f$.
This is formalized by the notion of the \emph{deepest down cut} of~$e$:

\begin{definition}[deepest down cut]
  \label{def:deterministic_deepest_cut}
  For a~tree edge $e$ in $\Tt$, we define the \emph{deepest down cut} of~$e$, denoted $\DeepestDnCut(e)$, as the deepest tree edge $f \geq e$ for which all back edges leaping over $e$ originate from $\Tt_f$.
\end{definition}

\begin{lemma}
  \label{lem:deterministic_deepest_cut_def}
  For every tree edge $e$, $\DeepestDnCut(e)$ is defined correctly and uniquely. Moreover, every $3$-edge cut $\{e, f, g\}$ of the considered kind satisfies $e < f \leq \DeepestDnCut(e)$.
  \begin{proof}
    Let $\Ss_e$ be the set of the tails of all the edges leaping over~$e$.
    Then, $\DeepestDnCut(e)$ is the deepest tree edge $f=uv$ such that $\Ss_e \subseteq V(\Tt_v)$.
    Equivalently, $v$ must be an ancestor of all the vertices in $\Ss_e$; hence, the deepest such vertex, the \emph{lowest common ancestor} of $\Ss_e$, is defined uniquely.
    The correctness of the definition of $\DeepestDnCut(e)$ follows.
    For the latter part of the lemma, observe that $f \leq \DeepestDnCut(e)$ if and only if every edge leaping over~$e$ originates from $\Tt_f$.
  \end{proof}
\end{lemma}

\begin{lemma}
  \label{lem:deterministic_deepest_cut_lca}
  For every tree edge $e$, $\DeepestDnCut(e) = uv$ is the tree edge whose head $v$ is the lowest common ancestor of two vertices: the tails of $\mindn[1](e)$ and $\maxdn[1](e)$.
  \begin{proof}
  Let $\Ss_e$ be the set defined in the proof of Lemma~\ref{lem:deterministic_deepest_cut_def}.
Then, $\DeepestDnCut(e) = uv$ is the tree edge whose head $v$ is the lowest common ancestor of $\Ss_e$.
  Equivalently, $v$ is the lowest common ancestor of two vertices: one of minimum preorder in $\Ss_e$, and the other of maximum preorder in $\Ss_e$.
  By definition, those two vertices are exactly the tails of $\mindn[1](e)$ and $\maxdn[1](e)$, respectively.
  \end{proof}
\end{lemma}

Thanks to Lemma~\ref{lem:deterministic_deepest_cut_lca}, we can compute $\DeepestDnCut(e)$ in constant time for each tree edge $e$ as the lowest common ancestor of two vertices can be computed in constant time after linear preprocessing~\cite{DBLP:journals/siamcomp/HarelT84}.

We now shift our focus to the lower tree edge $f$.
As in Section~\ref{sec:randomized_3cut}, the only possible candidate $g$ for a~back edge of the cut leaping over $f$ is given by $\maxup[1](f)$.
The only problematic part is locating the remaining tree edge $e$.
Previously, we utilized randomness in order to calculate the compressed hash of $e$ given the compressed hashes of $f$ and $g$.
Here, we instead use the following fact:

\begin{lemma}
  \label{lem:deterministic_deepest_edge_only}
  If $\{e, f, g\}$ is a~$3$-edge-cut of the considered kind for some tree edges $e$ and $f > e$, then $e$ is the deepest tree edge satisfying $\DeepestDnCut(e) \geq f$.
  \begin{proof}
    Naturally, the condition $\DeepestDnCut(e) \geq f$ must be satisfied if $\{e, f, g\}$ is a~$3$-edge-cut (Lemma~\ref{lem:deterministic_deepest_cut_def}).
    It only remains to prove that $e$ is the deepest edge with this property.
    Pick two tree edges $e_1$, $e_2$ such that $e_1 < e_2 < f$, both satisfying the condition regarding their deepest down cuts. We then partition the tree into four parts: $\Tt \setminus \Tt_{e_1}$, $\Tt_{e_1} \setminus \Tt_{e_2}$, $\Tt_{e_2} \setminus \Tt_f$, and $\Tt_f$. By the assumptions, all back edges leaping over $e_1$ or over $e_2$ must originate from $\Tt_f$ (Figure~\ref{fig:deterministic_only_deepest_dn_cut}).

    Since $G$ is $3$-edge-connected, $\{e_2, f\}$ cannot be its edge cut; there must exist a~back edge $x_1$ connecting $\Tt_{e_2} \setminus \Tt_f$ with the rest of the graph.
    It cannot leap over $e_2$ as it would also need to originate from $\Tt_f$; hence, $x_1$ connects $\Tt_f$ with $\Tt_{e_2} \setminus \Tt_f$.
    Analogously, $\{e_1, e_2\}$ is not an~edge cut of $G$, from which we infer that there exists another back edge $x_2$ connecting $\Tt_f$ with $\Tt_{e_1} \setminus \Tt_{e_2}$.
    Since we produced two separate back edges $x_1$, $x_2$ connecting $\Tt_f$ with $\Tt_{e_1} \setminus \Tt_f$, there cannot exist any $3$-edge-cut containing both $e_1$ and $f$.
    The statement of the lemma follows.
  \begin{figure}[h]
    \centering
    \includegraphics[scale=0.4]{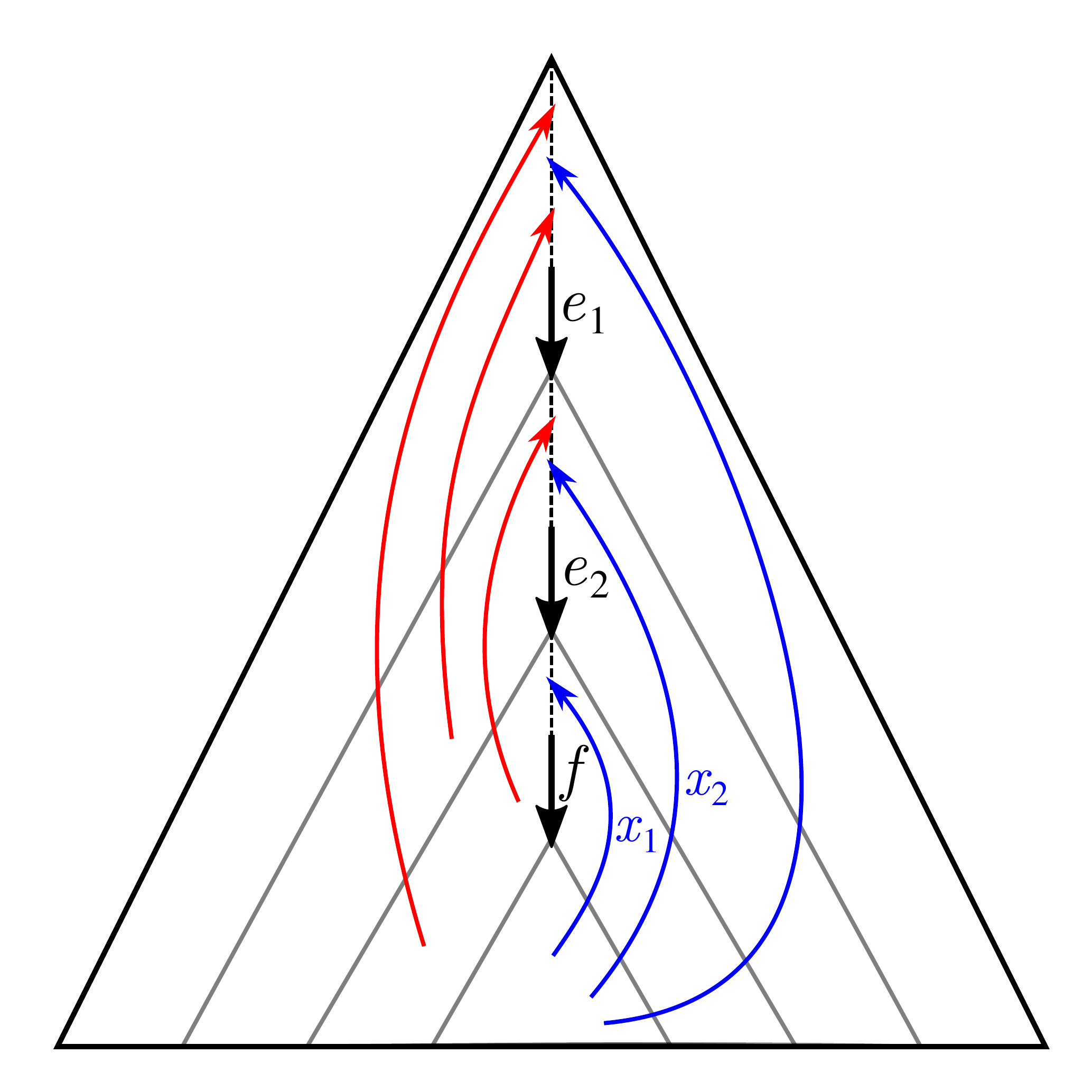}
    \caption{The setting in Lemma~\ref{lem:deterministic_deepest_edge_only}. Given that $\DeepestDnCut(e_1) \geq f$ and $\DeepestDnCut(e_2) \geq f$, the red back edges are forbidden from occurring in $G$, while the blue back edges can appear in $G$. We prove that both $x_1$ and $x_2$ appear in the graph.}
    \label{fig:deterministic_only_deepest_dn_cut}
  \end{figure}
  \end{proof}
\end{lemma}

Lemmas~\ref{lem:deterministic_deepest_cut_def} and~\ref{lem:deterministic_deepest_edge_only} naturally lead to the following idea: given a~tree edge $e$, the set of edges $f$ satisfying $\DeepestDnCut(e) \geq f$ is a~path $P_e$ connecting the head of $e$ with the head of $\DeepestDnCut(e)$; we assign this path a~weight $w_e$ equal to the depth of $e$ in~$\Tt$.
We then invoke Theorem~\ref{thm:min_edge_path} with the tree $\Tt$, the weighted paths $\{P_e\,\mid\,e \in E(\Tt)\}$, and $k = 1$.
This allows us to find, in $\Oh(n)$ time, for each tree edge $f \in E(\Tt)$, the path $P_e$ of the maximum weight containing $f$ as an~edge.
This path naturally corresponds to the tree edge~$e$ from Lemma~\ref{lem:deterministic_deepest_edge_only}.
This way, for every tree edge $f$, we have uniquely identified a~back edge $g$ and a~tree edge $e$ such that the only possible $3$-edge-cut containing $f$ as a~deeper tree edge, if it exists, is $\{e, f, g\}$.

It only remains to verify that $\{e, f, g\}$ is a $3$-edge-cut.
We need to make sure that:
\begin{itemize}
  \item All back edges leaping over $e$ originate from $\Tt_f$.
    The equivalent condition $\DeepestDnCut(e) \geq f$ is guaranteed by the algorithm, so we do not need to check it again.
  \item Exactly one edge ($g$) connects $\Tt_f$ with $\Tt_e \setminus \Tt_f$.
    Since $g = \maxup[1](f)$, it suffices to verify that the edge $\maxup[2](f)$, which is a~back edge leaping over $f$ whose head is the deepest apart from $g$, also leaps over $e$ (i.e., it does not terminate in $\Tt_e$).
\end{itemize}

It can be easily verified that the implementation of the subroutine is deterministic and runs in linear time with respect to the size of $G$.

\subsection{Two tree edges, upper case}
\label{ssec:deterministic_2edges_up}

Recall that in this case, we are required to find all $3$-edge-cuts intersecting $\Tt$ in two edges $e$ and $f$, such that $e < f$ and the remaining back edge $g$ connects $\Tt_e \setminus \Tt_f$ with $\Tt \setminus \Tt_e$ (Figure \ref{fig:randomized-two-tree-edge-cases}). %\todo{figure from random}
Similarly to the randomized case, we use the fact that given a~tree edge $e$, the tail of $g$ has either the smallest preorder (i.e., $g = \mindn[1](e)$) or the largest preorder (i.e., $g = \maxdn[1](e)$) among all the back edges leaping over $e$.
Without loss of generality, assume that $g = \mindn[1](e)$; the latter case is analogous.

In a~similar vein to previous case, observe that if $\{e, f, g\}$ is a~$3$-edge-cut for some $f > e$, then all back edges leaping over $e$ other than $g$ must originate from $\Tt_f$.

This leads to the following slight generalization of $\DeepestDnCut$ (Definition~\ref{def:deterministic_deepest_cut}):
\begin{definition}
  \label{def:deterministic_deepest_nomin_cut}
  For a~tree edge $e$ in $\Tt$, we define the value $\DeepestDnCutNoMin(e)$ as the deepest tree edge $f \geq e$ for which all back edges leaping over $e$ other than $\mindn[1](e)$ originate from $\Tt_f$.
\end{definition}
The process of computation of $\DeepestDnCut(e)$ asserted by Lemma~\ref{lem:deterministic_deepest_cut_lca} can be easily modified to match our needs: recall that $\DeepestDnCut(e)$ is the tree edge whose head was the lowest common ancestor of the origins of $\mindn[1](e)$ and $\maxdn[1](e)$.
Since we exclude $\mindn[1](e)$ from the condition in Definition~\ref{def:deterministic_deepest_nomin_cut}, we use the edge $\mindn[2](e)$ instead: that is, the back edge leaping over $e$, with the next lowest preorder of the tail.
Therefore, we compute $\DeepestDnCutNoMin(e)$ as the tree edge whose head is the lowest common ancestor of the tails of $\mindn[2](e)$ and $\maxdn[1](e)$ (Figure~\ref{fig:deterministic_upcase_f0}).

\begin{figure}[h]
\centering
  \includegraphics[scale=0.4]{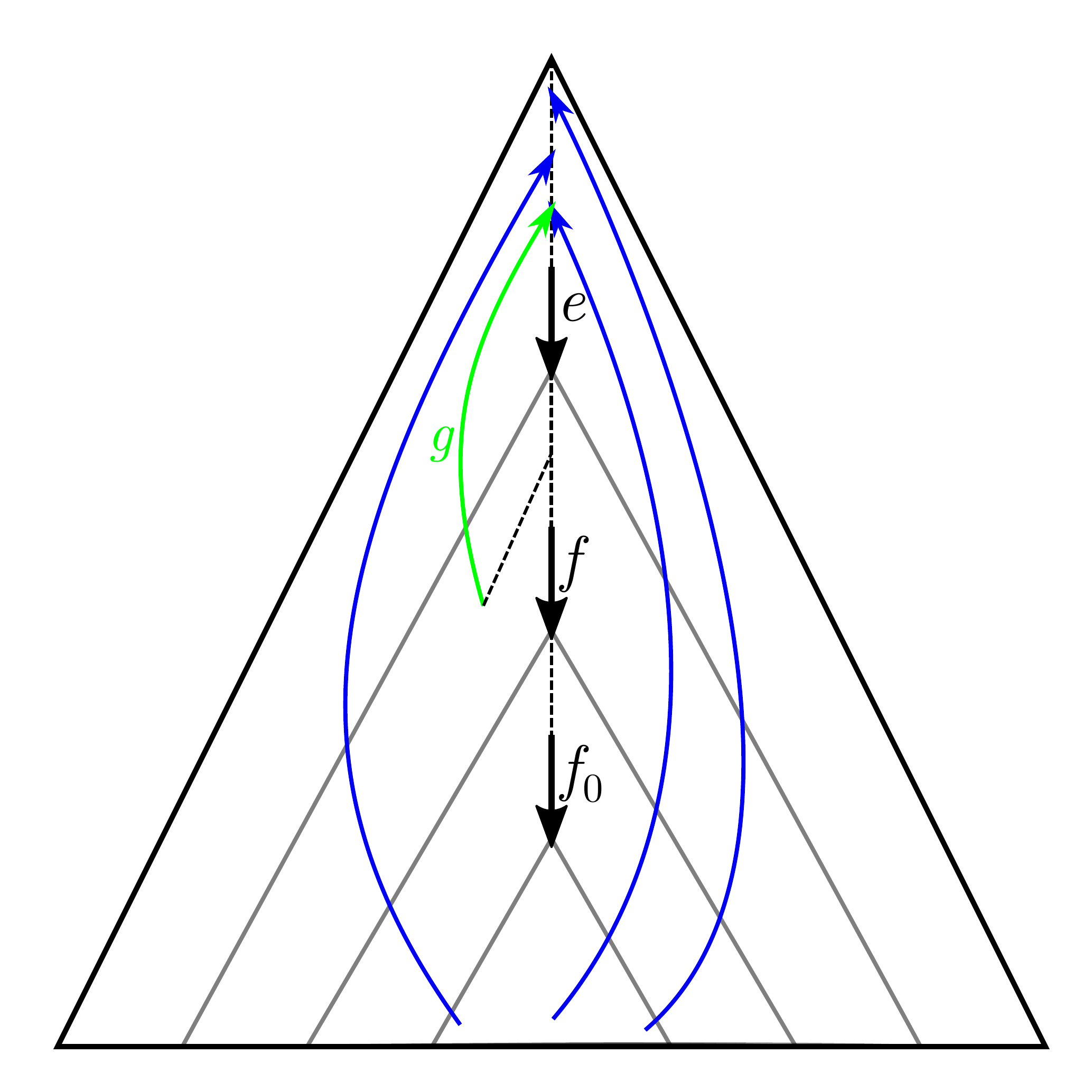}
\caption{The computation of $f_0 = \DeepestDnCutNoMin(e)$ excludes $g = \mindn[1](e)$, and computes the lowest common ancestor of the remaining back edges leaping over $e$.}
\label{fig:deterministic_upcase_f0}
\end{figure}

Now, fix the shallower tree edge $e$ of the cut.
Let $f_0 := \DeepestDnCutNoMin(e)$.
Then, if some tree edge $f > e$ belongs to the $3$-edge-cut $\{e, f, g\}$, then $f \leq f_0$ (otherwise, there would be multiple edges connecting $\Tt_e \setminus \Tt_f$ with $\Tt \setminus \Tt_e$).
The natural question is then: is $\{e, f_0, g\}$ a~$3$-edge-cut?
The only possible problem is that there may exist back edges connecting $\Tt_{f_0}$ with $\Tt_e \setminus \Tt_{f_0}$.
Fortunately, if this is the case, then the back edge $\maxup1(f_0)$, dependent only on $f_0$, is one of these back edges.

Hence, let $h_0 := \maxup[1](f_0)$.
If $h_0$ leaps over $e$, we are done, and $\{e, f_0, g\}$ is an~edge cut.
Otherwise, the subtree $\Tt_f$ for the sought $3$-edge-cut $\{e, f, g\}$ must contain the head of $h_0$ (or else $h_0$ would connect $\Tt_f$ with $\Tt_e \setminus \Tt_f$).
Let then $f_1$, $e \leq f_1 < f_0$ be the deepest tree edge containing the head of $h_0$, that is, the edge whose head coincides with the head of $h_0$.
Then, the edge $f$ of the $3$-edge-cut $\{e, f, g\}$ must satisfy $e < f \leq f_1$.
We can then repeat this procedure: given $f_1$, we compute $h_1 := \maxup[1](f_1)$, and either $h_1$ leaps over $e$ and we are done, or we calculate another edge $f_2 < f_1$ supplying a~better bound on the depth of $f$ (Figure~\ref{fig:deterministic_upcase_iter}).
Since the graph is finite, this process terminates in $k = \Oh(n)$ steps, producing the lowest tree edge $f_k \geq e$ such that no back edge connects $\Tt_{f_k}$ with $\Tt_e \setminus \Tt_{f_k}$, and no back edge other than $g$ connects $\Tt_e \setminus \Tt_{f_k}$ with $\Tt \setminus \Tt_e$.
If $f_k = e$, then no $3$-edge-cut $\{e, f, g\}$ exists for $f > e$.
Otherwise, $\{e, f_k, g\}$ is naturally a~correct $3$-edge-cut.
Since two edges of a $3$-edge-cut uniquely identify the third edge (Lemma~\ref{lem:randomized_unique_remaining_edge}), we conclude that this is the only $3$-edge-cut containing $e$ and $g$.

\begin{figure}[h]
\centering
  \includegraphics[scale=0.4]{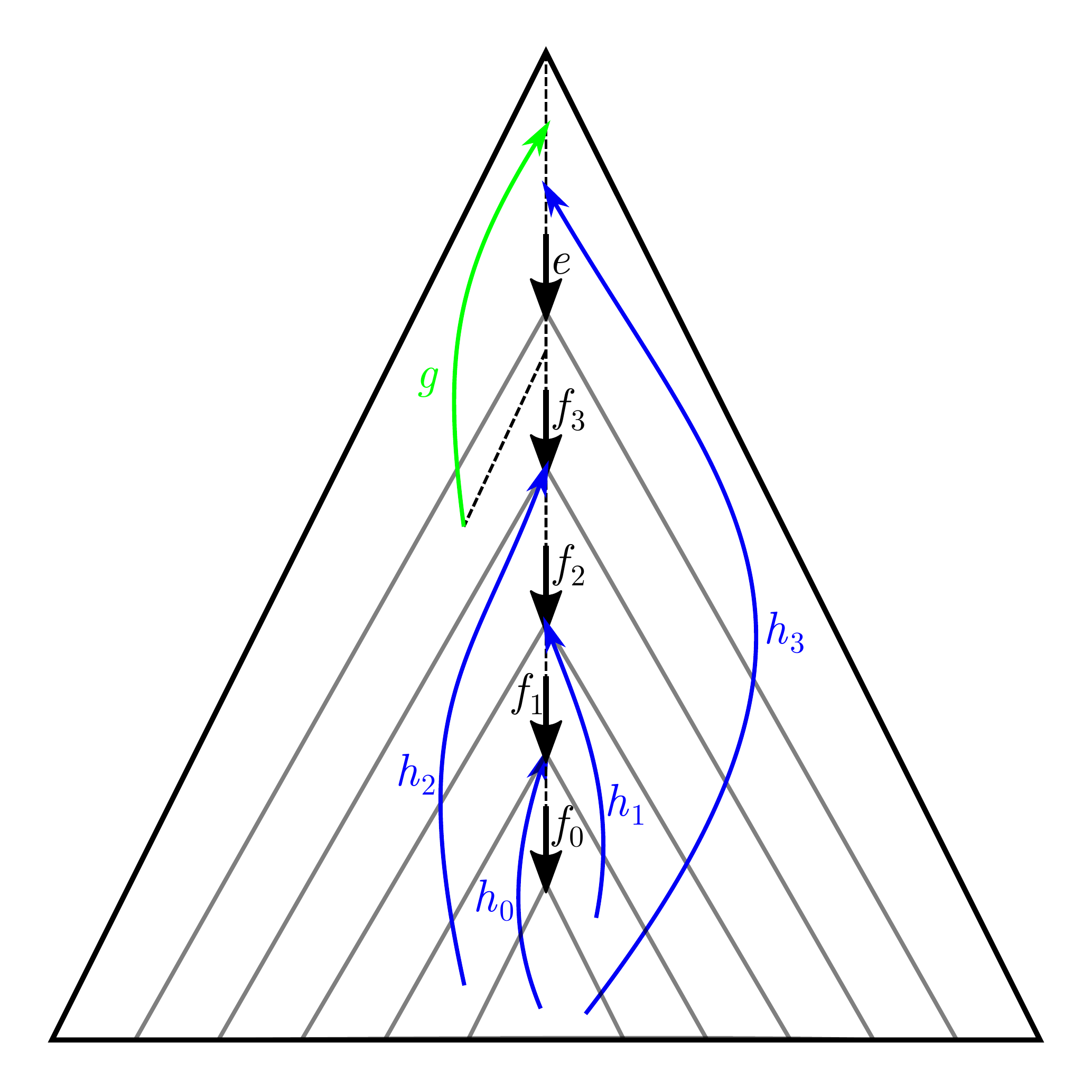}
\caption{The main iteration of the algorithm. We set $f_0 = \DeepestDnCutNoMin(e)$, $h_i = \maxup[1](f_i)$, and $f_{i+1}$ is the tree edge with the same head as $h_i$.}
\label{fig:deterministic_upcase_iter}
\end{figure}

We can now sum up the inefficient deterministic implementation of the algorithm for the current case (Algorithm~\ref{alg:deterministic_2edges_up_bad}).

\begin{algorithm}[H]
  \begin{algorithmic}
  \caption{Inefficient deterministic solution for ``two tree edges, upper case''}
  \label{alg:deterministic_2edges_up_bad}
  \Function {SolveTwoTreeEdgesUpper} {}
  \For{$e \in E(\Tt)$}
    \State $f_0 \gets \DeepestDnCutNoMin(e)$
    \State $g \gets \mindn[1](e)$
    \State $i \gets 0$
    \While {$f_i \neq e$ and $\maxup[1](e).\mathrm{head} \in \Tt_e$}
      \State $h_i \gets \maxup[1](e)$
      \State $f_{i+1} \gets$ the tree edge with the same head as $h_i$
      \State $i \gets i + 1$
    \EndWhile
    \If {$f_i \neq e$}
      \State add $\{e, f_i, g\}$ to the list of $3$-edge-cuts
    \EndIf
  \EndFor
  
  \Comment {{\small Run the analogous algorithm for $\DeepestDnCutNoMax$ instead of $\DeepestDnCutNoMin$}.}
  \EndFunction

  \end{algorithmic}
\end{algorithm}

%\begin{enumerate}
%  \item For every tree edge $e$, in any order:
%  \begin{enumerate}
%    \item $f_0 := \DeepestDnCutNoMin(e)$, $i = 0$.
%    \item $g := \mindn[1](e)$.
%    \item while $f_i \neq e$ and the destination of $\maxup[1](f_i)$ belongs to $\Tt_e$:
%    \begin{enumerate}
%      \item $h_i := \maxup[1](f_i)$
%      \item $f_{i+1} :=$ the tree edge with the same destination as $h_i$
%      \item $i := i + 1$
%    \end{enumerate}
%    \item If $f_i \neq e$: output $\{e, f_i, g\}$ as a~cut.
%  \end{enumerate}
%  \item Do the analogous algorithm, but for $g = \maxdn[1](e)$.
%\end{enumerate}

In order to optimize the algorithm, we seek to eliminate the $\mathsf{while}$ loop.
Indeed, for each $e \in E(\Tt)$, the process is very similar: start with some edge $f > e$, and then repeatedly replace $f$ with a~higher edge, until a~replacement would result in an~edge closer to the root than $e$.
We can model this process with a~rooted tree $\Uu$ defined as~follows:
\begin{itemize}
  \item the vertices of $\Uu$ are the tree edges of $\Tt$, and $\bot$,
  \item $\bot$ is the root of $\Uu$,
  \item for a~non-root vertex $e$ of $\Uu$, its parent is the tree edge with the same head as $\maxup[1](e)$.
    If the head of $\maxup[1](e)$ coincides with the root of $\Tt$, then the parent of $e$ in $\Uu$ is $\bot$.
\end{itemize}
Naturally, $\Uu$ can be constructed in linear time with respect to the size of $\Tt$; moreover, if an~edge $p$ is a~parent of another edge $q$ in $\Uu$, then $p <_\Tt q$.

Now, the $\mathsf{while}$ loop is equivalent to repeated replacement of $f_0$ with its parent in $\Uu$ as long as the parent is greater or equal than $e$ with respect to $<_\Tt$.
In other words, the final edge $f' := \mathsf{final}(f_0, e)$ is taken as the shallowest ancestor of $f_0$ in $\Uu$ for which $f' \geq_\Tt e$.

This leads to the final idea: we simulate a~forest of rooted subtrees of $\Uu$ using a~disjoint set union data structure $F_\Uu$ (Theorem~\ref{thm:linear_fu}).
Each subtree additionally keeps its root, which can be retrieved from $F_\Uu$ (Lemma~\ref{enrich-parent}).
Initially, each subtree of $\Uu$ contains a~single vertex.

After the initialization of $F_\Uu$, we iterate $e$ over the tree edges of $\Tt$ in the decreasing order of depth in $\Tt$.
Throughout the process, we maintain the following invariant on $F_\Uu$: an edge $ef \in E(\Uu)$ for $e < f$ has been added to the forest if and only if $e$ has been considered at any previous iteration of the main loop.
Hence, at the beginning of the iteration for a~given edge $e$, we add to $F_\Uu$ all tree edges of $\Uu$ originating from $e$.
At this point of time, for every tree edge $f$ such that $f >_\Tt e$, the edge $\mathsf{final}(f, e)$ is given by $F_\Uu.\mathsf{lowest}(f)$.
This reduces the entire $\mathsf{while}$ loop to a~single $\mathsf{lowest}$ query on $F_\Uu$ (Algorithm~\ref{alg:deterministic_2edges_up_good}).
%\todo{smu:zrobić jakoś twardy enter aby nie wrzucało innego tekstu między ':' i algosa, albo zamienic ':' na referencje do labelki algosa}
\begin{algorithm}[H]
  \begin{algorithmic}
  \caption{Efficient deterministic solution for ``two tree edges, upper case''}
  \label{alg:deterministic_2edges_up_good}
  \Function {SolveTwoTreeEdgesUpper} {}
  \State $\Uu \gets$ the rooted tree on $E(\Tt) \cup \{\bot\}$ defined above
  \State $F_\Uu \gets$ the disjoint set union data structure built on $\Uu$ (Lemma~\ref{lem:linear_fu_with_lowest})
  \For{$e \in E(\Tt)$, in order from the deepest to the shallowest}
    \For{$c$ -- child of $e$ in $\Uu$}
      \State $F_\Uu.\mathsf{union}(e, c)$
    \EndFor
    \State $g \gets \mindn[1](e)$
    \State $f_0 \gets \DeepestDnCutNoMin(e)$
    \State $f' \gets F_\Uu.\mathsf{lowest}(f_0)$
    \If {$f' \neq e$}
      \State add $\{e, f', g\}$ to the list of $3$-edge-cuts
    \EndIf
  \EndFor

  \Comment {{\small Run the analogous algorithm for $\DeepestDnCutNoMax$ instead of $\DeepestDnCutNoMin$}.}
  \EndFunction

  \end{algorithmic}
\end{algorithm}

%\begin{enumerate}
%  \item $\Uu$ := the rooted tree on $E(\Tt) \cup \{\bot\}$ defined above.
%  \item $F_\Uu$ := the disjoint set union data structure built on $\Uu$ (Theorem~\ref{thm:linear_fu}).
%  \item For every tree edge $e$ of $\Tt$, in order from the deepest to the shallowest:
%  \begin{enumerate}
%    \item For every child $c$ of $e$ in $\Uu$: $F_\Uu.\mathsf{union}(e, c)$.
%    \item $g := \mindn[1](e)$.
%    \item $f_0 := \DeepestDnCutNoMin(e)$.
%    \item $f_i := F_\Uu.\mathsf{lowest}(f_0)$.
%    \item If $f_i \neq e$: output $\{e, f_i, g\}$ as a~cut.
%  \end{enumerate}
%  \item Apply the analogous algorithm (steps $1$--$3$) for $g = \maxdn[1](e)$.
%\end{enumerate}

It can be easily seen that we initialize $F_\Uu$ on a~tree with $n$ vertices, and we issue $\Oh(n)$ queries to it in total.
Therefore, the whole subroutine runs in time linear with respect to the size of $G$.

Summing up, we replaced each randomized subroutine with its deterministic counterpart, preserving the linear guarantee on the runtime of the algorithm.
We conclude that there exists a deterministic linear-time algorithm listing $3$-edge-cuts in $3$-edge-connected graphs.

\section{Reconstructing the structure of $4$-edge-connected components}
\label{Smu2}

In this section, we show how to build a structure of $4$-edge-connected components of a~$3$-edge-connected graph $G$, given the set $\Cc$ of all $3$-edge-cuts in $G$.

First of all, recall what such a~structure looks like.

\begin{theorem}\cite[Corollary 8]{SimplerCactus}
\label{thm:treerec_existence}
For a~$3$-edge-connected graph $G = (V,E)$, there exists a tree $H = (U,F)$ along with functions $\phi: \Cc \rightarrow F$ and $\psi:V \rightarrow U$,
such that $\phi$ is a bijection from the $3$-edge-cuts of $G$ to the edges of $H$,
and $\psi$ maps (not necessarily surjectively) the vertices of $G$ to the vertices of $H$
in such a way that the whole $4$-edge-connected components are mapped to the same vertex.

Moreover, if a~$3$-edge-cut $c$ partitions the vertices of $G$ into two parts $V_1$ and $V_2$, then $\phi(c)$ partitions the vertices of $H$ into $U_1$ and $U_2$
such that $\psi^{-1}(U_1) = V_1$ and $\psi^{-1}(U_2) = V_2$.
\end{theorem}

Figure~\ref{fig:treerec_example} shows an~example of decomposition postulated by Theorem~\ref{thm:treerec_existence}.

\begin{figure}[h]
  \centering
  \includegraphics[scale=0.6]{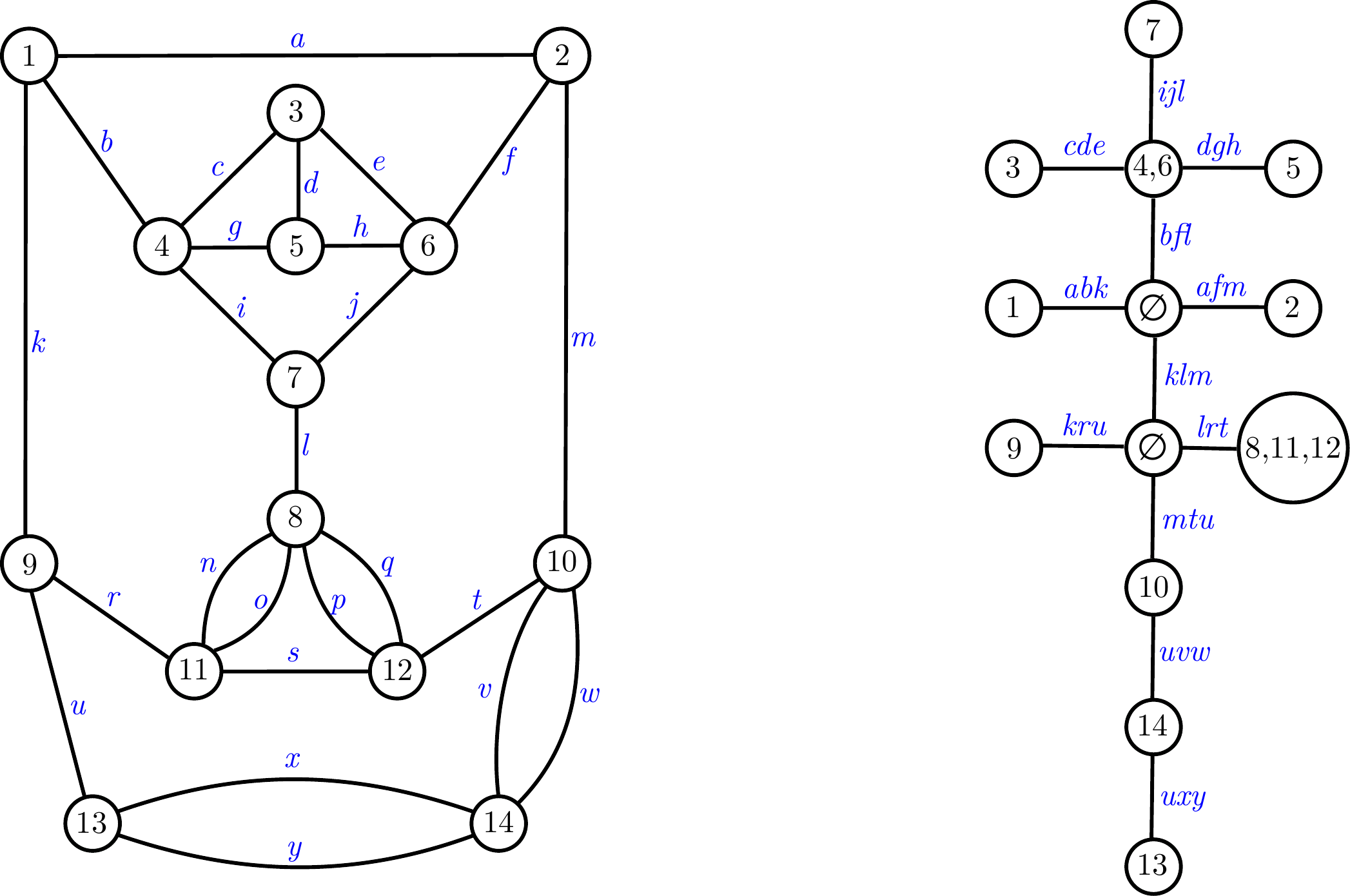}
  \caption{A~$3$-edge-connected graph with $14$ vertices and $25$ edges (left) and its tree $H$ (right).}
  \label{fig:treerec_example}
\end{figure}

The tree $H$ is usually unrooted in the literature.
However, we are going to root it.
Namely, we take a~depth-first search tree $\Tt$ of $G$, rooted at some vertex $r$, and we root $H$ at $\psi(r)$.
For a~vertex $u \in U$, we let $H_u$ denote the subtree of $H$ rooted at $u$.

\begin{definition}
  For a $3$-edge-cut $c$, we define $P(c)$ as the set of vertices from the connected component of
  $G \setminus c$ not containing $r$.
\end{definition}
  We remark that since $c$ is a~minimal cut, $G \setminus c$ consists of two connected components, so $P(c)$ is determined uniquely.

\begin{lemma}\label{lem:order_of_P}
  Let $v,u \in U$, and let $e_1, e_2, \ldots, e_k \in F$ be the sequence of edges of $H$ on the path from $v$ and $u$ in $H$. 
  If $v$ is an ancestor of $u$, then for each pair of integers $i$, $j$ such that $1 \leq i < j \leq k$, we have $|P(\phi^{-1}(e_i))| > |P(\phi^{-1}(e_j))|$.
\end{lemma}
\begin{proof}
  It suffices to only consider the inequalities for the consecutive edges from the considered path.
  To this end, let us fix some $i \in [1, k - 1]$ and consider the edges $(x,y) := e_i$ and $(y,z) := e_{i+1}$.
  By the definition of $H$, we have $P(\phi^{-1}(e_i)) = \psi^{-1}(H_y)$ 
  since both sides describe the connected component of $G \setminus \phi^{-1}(e_i)$ not containing $r$.
  Similarly, $P(\phi^{-1}(e_{i+1})) = \psi^{-1}(H_z)$.
  As $H_z$ is a~subtree of $H_y$, we infer that
  $$P(\phi^{-1}(e_{i+1})) = \psi^{-1}(H_z) \subseteq \psi^{-1}(H_y)= P(\phi^{-1}(e_i)).$$
  This implies that $|P(\phi^{-1}(e_i))| \ge |P(\phi^{-1}(e_{i+1}))|$.
  However, observe that each set $P(c)$ uniquely defines a~cut in the graph $G$ as the set of all edges between $P(c)$ and $V \setminus P(c)$.
  By the definition of $H$, we have $\phi^{-1}(e_i) \neq \phi^{-1}(e_{i+1})$.
  Therefore, $P(\phi^{-1}(e_i)) \neq P(\phi^{-1}(e_{i+1}))$, which finishes the proof of the lemma.
\end{proof}

\begin{lemma}\label{lem:part_sizes}
  Given a $3$-edge-connected graph $G$ and the set of all its $3$-edge-cuts $\Cc$, the sizes of $P(c)$ for all $c \in \Cc$ can be computed in linear time with respect to the size of $G$.
\end{lemma}
\newcommand{\cntOnPath}{{\rm cntOnPath}}
\newcommand{\deeper}{{\rm deeper}}
\begin{proof}
  Let us consider a $3$-edge-cut $c$.
  If two vertices are connected by an edge $(u, v)$ of the cut, then exactly one of $u$, $v$ is contained in $P(c)$.
  Similarly, if an edge $(u,v)$ is not contained in $c$, then $u \in P(c) \Leftrightarrow v \in P(c)$.
  The root $r$ is by definition not in $P(c)$.
  Therefore, to decide whether a vertex $v$ belongs to $P(c)$, it is enough to check the parity of the number of edges of the cut
  on the path from $r$ to $v$ in $\Tt$.
  Precisely, let us define $\cntOnPath(c, v) := \sum_{e \in c \cap \Tt} [v \in \Tt_e]$ as the number of the tree edges of the cut $c$ on the path from the root of $\Tt$ to a~vertex $v \in V(\Tt)$, and $\deeper(e)$ as the deeper endpoint of $e$ in $\Tt$.
  Then, a~vertex $v$ is disconnected from the root $r$ by the $3$-edge-cut $c$ if there is an~odd number of the edges of $c$ on the path from $r$ to $v$ in $\Tt$:
  $$|P(c)| = \sum_{v \in V} [2 \nmid \cntOnPath(c,v)] = \sum_{e \in \Tt \cap c} (-1)^{1 + \cntOnPath(c, \deeper(e))} \cdot |\Tt_e|.$$
  Assuming the sizes of the subtrees of $\Tt$ have been precomputed, the value of $P(c)$ can be computed for a~particular cut $c$ in constant time.
  The total time complexity of this algorithm is $\Oh(m + |\Cc|)$,
  which is linear in the size of $G$ as $|\Cc| \leq \Oh(n)$.
\end{proof}

To reconstruct $H$, we will also use the following structural lemma about cuts sharing the same edge of graph $G$.
%\todo{try finding this in literature}
\begin{lemma}\label{lem:edges_are_paths}
  For an edge $(u,v) = e \in E$, let $l(e)$ be the set of all $3$-edge-cuts containing $e$.
  The image $\phi(l(e))$, i.e., the set of all edges of $H$ corresponding to the edge cuts containing $e$, forms a path in $H$ between $\psi(u)$ and $\psi(v)$.
\end{lemma}
\begin{proof}
  Let us consider the vertices $t_u := \psi(u)$ and $t_v := \psi(v)$ of $H$, and pick an~edge $f \in F$.
  If $f$ lies on the path between $t_u$ and $t_v$ in $H$, then 
  $f$ divides $H$ into two parts $U_1$ and $U_2$ such that $t_u \in U_1$ and $t_v \in U_2$.
  Thus, $f$ corresponds to the cut $\phi^{-1}(f)$ separating $u$ from $v$, which must necessarily contain $e=(u,v)$ as an~edge.
%  so cut $\phi^{-1}(f)$ separates vertices $u$ and $v$
%  and $e \in \phi^{-1}(f)$.
  
  Similarly, if $f$ does not lie on the path between $t_u$ and $t_v$ in $H$, then
  $f$ divides $H$ into two parts $U_1$ and $U_2$ such that $t_u,t_v \in U_1$ and $t_u,t_v \not\in U_2$.
  Therefore, the cut $\phi^{-1}(f)$ does not separate $u$ from $v$.
  Since this cut is minimal, we get that $e \not\in \phi^{-1}(f)$.
  This completes the proof.
\end{proof}
  We remark an~edge case in Lemma~\ref{lem:edges_are_paths}: if $\psi(u) = \psi(v)$ for some edge $e=(u, v)$, then the image $\phi(l(e))$ is empty.
  Moreover, since each $3$-edge-cut $c \in \Cc$ contains at least one tree edge of $\Tt$, each edge $\phi(c)$ is covered by at least one path $\phi(l(e))$ for $e \in E(\Tt)$.
  Equivalently, $H$ is a~tree, rooted at $\psi(r)$, equal to the union of all paths $\phi(l(e))$ for $e \in E(\Tt)$.
  This representation of $H$ is the cornerstone of our algorithm reconstructing $H$ from $G$ and $\Cc$.

%We are now ready to give the description of the algorithm constructing the structure of the $4$-connected components of $G$.
%The gist of our algorithm is to run two parallel traversals
%visiting corresponding places in $\Tt$ and $H$ defined by mappings $\phi$ and $\psi$.
%One of them simply traverses $\Tt$ in DFS order,
%the second one traverses and builds initially unknown tree $H$ using information about given $3$-edge cuts.\todo{not anymore}

\begin{lemma}
  There exists a linear time algorithm which, given a~graph $G$ and the list $\Cc$ of all 3-edge-cuts of $G$, constructs the tree $H$, along with the mappings $\phi$ and $\psi$.
\end{lemma}

\begin{proof}

  First, for each edge $e \in \Tt$, one can create a list $l(e)$ of all $3$-edge-cuts containing $e$.
  We sort each such list decreasingly by the size of $P(c)$; this can be done using radix sort in linear time.

  Lemma~\ref{lem:part_sizes} guarantees an~important property of each sorted list of cuts: if two cuts $c_1, c_2$ of the same list map to the edges $e_1$, $e_2$ in $H$, respectively, such that $e_1 <_H e_2$, then $c_1$ appears earlier in the list than $c_2$.

  Let $e_1, e_2, \dots, e_{n-1} \in E(\Tt)$ be the sequence of edges visited by a~depth-first search of $\Tt$.
  We note that this sequence of edges is consistent with the tree order $<_\Tt$; i.e., if $e_i <_\Tt e_j$ for some two edges in the sequence, then $i < j$.
  Let $e_i = u_i v_i$, where $u_i$ is the vertex of $\Tt$ closer to the root~$r$.

  We will create the sought tree $H$ iteratively.
  Initially, take $H$ as the~tree containing a~single vertex $\psi(r)$.
  We consider the edges $e_1, e_2, \dots, e_{n-1}$ of $\Tt$ in this order, maintaining the following invariant after $k$ iterations of the algorithm:
  
%  \medskip
  
  \paragraph{Invariant.} $H$ is a~connected tree, rooted at $\psi(r)$, equal to the union of all paths $\phi(l(e_i))$ for $i \in \{1, 2, \dots, k\}$.
  
  \medskip
  
%  \begin{itemize}
%    \item $\Ee = \{r, v_1, v_2, \dots, v_k\}$, i.e., $\Ee$ is the set of vertices reachable from $r$ in $\Tt$ using the first $k$ edges in the sequence, and
%    \item 
%  \end{itemize}
  
  It is clear that after $n - 1$ iterations, $H$ will be the rooted tree representing the structure of all $3$-edge-cuts in $G$.
  
  Consider the $k$-th iteration of the algorithm, $k \in \{1, 2, \dots, n - 1\}$.
  In this iteration, we need to add to $H$ the path $\phi(l(e_k))$, which originates from $\psi(u_k)$ and terminates at $\psi(v_k)$.
  To this end, we first notice that $\psi(u_k)$ already exists as a~vertex of $H$: either $u_k = r$, which means that $\psi(u_k)$ is the root of $H$; or $u_k$ is the head $v_t$ of some edge $e_t$ earlier in the order, which implies that $\psi(u_k) = \psi(v_t)$ exists in $H$ as one of the endpoints of $\phi(l(e_t))$.
  Hence, some prefix of $\phi(l(e_k))$ already exists in $H$, and it only remains to add the suffix of this path to $H$.

  The considered path consists of two vertical parts: the first, from $\psi(u_k)$ to the lowest common ancestor of $\psi(u_k)$ and $\psi(v_k)$ towards the root $\psi(r)$, and the second, from the lowest common ancestor to $\psi(v_k)$ away from the root of $H$.
  Fortunately, the first part of the path is already included in $H$: since the tree is rooted, the entire vertical path from the root $\psi(r)$ to $\psi(u_k)$ is present in $H$.
  Then, the second part of the path is clearly formed by the edges of $l(e)$ not present in the first part of the path, sorted in the same order as $l(e)$ (Lemma~\ref{lem:order_of_P}).
  
  Now, adding the path $\phi(l(e_k))$ to $H$ is rather straightforward: first, starting from $\psi(u)$, we go up the tree $H$ along the edges of $H$ corresponding to the edge cuts containing $e_k$.
  Then, we start going down the tree: we iterate the list $l(e_k)$ of cuts, excluding the cuts that correspond to the edges visited in the first part of the traversal.
  For each such cut, we go down the tree along the edge corresponding to the cut (creating it, if necessary).
  This is summed up by the following implementation of a~single iteration of the algorithm:

%  After computing sizes of $P(c)$ by lemma \ref{lem:part_sizes}
%  one can 
%  Afterwards, we can start parallel traverses of $\Tt$ and $H$.
%  Let us denote variables containing currently visited vertices $cur_G$ and
%  $cur_T$, respectively.
  
%  At the beginning $cur_G = r$ and $cur_T$ is a newly created vertex of $H$ 
%  representing $\psi(r)$.
%  Let $\Ee$ be a set of already traversed edges of $\Tt$.
%  After each step the following two invariants hold:
%  \begin{enumerate}
%    \item Already constructed edges of $H$ are exactly $\phi(l(\Ee))$.
%    \item $\psi(cur_G) = cur_T$.
%  \end{enumerate}

  \begin{algorithm}[H]
    \begin{algorithmic}
    \caption{A single iteration of the algorithm}
    \label{alg:AddPathToH}
    \Function {AddPathToH} {e} \Comment $e = (u, v)$ is an edge of $\Tt$
    \State $x \gets \psi(u)$
    
    \Comment Go up the tree along the first part of the path, marking the edges visited on the way.
    \While{$e \in \phi^{-1}({\rm EdgeToParent}(x))$}
      \State $\mathrm{touched}(\phi^{-1}({\rm EdgeToParent}(x))) \gets e$
      \State $x \gets \mathrm{parent}(x)$
    \EndWhile
    
    \Comment Go down the tree along the edges of $l(e)$ not visited by the first loop.
    \For{$c \in l(e)$}
      \If{$\mathrm{touched}(c) \neq e$}
        \If{$\phi(c) = \bot$}
          \State $\phi(c) \gets$ a fresh edge in $H$ from $x$ to a~new vertex, corresponding to the cut $c$
%          \rm{NewTEdge}(cur_T, \rm{NewTVertex}())$
        \EndIf
        \State \textbf{assert} the shallower end of $\phi(c)$ is equal to $x$
        \State $x \gets$ the deeper end of $\phi(c)$
      \EndIf
    \EndFor
    \State $\psi(v) \gets x$
    \EndFunction
    \end{algorithmic}
  \end{algorithm}

  We remark that each $3$-edge-cut $c$ is considered in only a~constant number of calls to \textsc{AddPathToH}: an~edge of $H$ corresponding to $c$ is traversed by $\text{\textsc{AddPathToH}}(e)$ only if $e \in c$.
  Therefore, the total time complexity of all calls to \textsc{AddPathToH} is $\Oh(m + |\Cc|) = \Oh(m + n)$.
\end{proof}

  This concludes the construction of a~tree representing all $3$-edge-cuts in $G$.
  As a~result, each vertex of $H$, as long as it is not empty, contains a~single $4$-edge-connected component of $G$.
  Hence, this algorithm also computes the decomposition of a~$3$-edge-connected graph $G$ into $4$-edge-connected components in total linear time.

\section{Open problems}
\label{open-problems}

As a natural open problem whose resolving would complement this result nicely, we suggest investigating whether it is possible to design an~algorithm computing 4-vertex-connected components in linear time.

\paragraph{Problem 1.} Given an undirected graph $G = (V, E)$, is it possible to find all 4-vertex-connected components of $G$ in linear time?

\bigskip

Additionally, it might be worth investigating whether this result can be lifted to higher connectivities.
Admittedly, this seems quite complicated as the algorithm presented in Sections~\ref{sec:randomized_3cut} and \ref{sec:deterministic_3cut} is crafted specifically for $4$-edge-connectivities.

\paragraph{Problem 2.} Given an undirected graph $G = (V, E)$, is it possible to find all 5-edge-connected components of $G$ in linear time?

\bigskip

We also remark that our algorithm assumes the word RAM model in which we can perform any arithmetic and bitwise operations on $\Oh(\log n)$-bit words in constant time; this is required by Theorem~\ref{thm:linear_fu}.
The natural question is whether this assumption can be avoided.

\paragraph{Problem 3.} Given an undirected graph $G = (V, E)$, is it possible to find all 4-edge-connected components of $G$ in linear time in the pointer machine model?

\bibliographystyle{abbrv}
\bibliography{4-connected-components.bib}

\end{document}